\documentclass[11pt]{article}

\usepackage{latexsym}
\usepackage{amssymb}
\usepackage{amsmath}
\usepackage{amsthm,fullpage}
\usepackage[noend]{algorithmic}
\usepackage{algorithm}
\usepackage{bm}

\usepackage{nth}

\usepackage{subfig}
\usepackage{graphicx}
\usepackage{enumitem}   

\usepackage{forloop}
\usepackage[paper=letterpaper, margin=1in]{geometry}
\usepackage{nicefrac}
\usepackage{rotating}
\usepackage{array,multirow}
\usepackage{threeparttable}
\usepackage{etoolbox}
\usepackage{tikz,tkz-graph}
\usetikzlibrary{arrows,shapes}
\usetikzlibrary{positioning,calc}
\usetikzlibrary{decorations.pathreplacing}
\usetikzlibrary{fadings}
\tikzset{>= stealth'}
\tikzstyle{vertex}=[circle, draw,fill=gray!20, inner sep=0pt, minimum size=16pt]
\tikzstyle{square}=[rectangle, draw,fill=gray!20, inner sep=0pt, minimum size=16pt]
\tikzstyle{svertex}=[circle, draw,fill=gray!20, inner sep=0pt, minimum size=12pt]
\tikzstyle{ssquare}=[rectangle, draw,fill=gray!20, inner sep=0pt, minimum size=12pt]
\tikzstyle{tvertex}=[circle, draw,fill=gray!20, inner sep=0pt, minimum size=8pt]
\tikzstyle{tsquare}=[rectangle, draw,fill=gray!20, inner sep=0pt, minimum size=6pt]

\usepackage[compact]{titlesec}
\usepackage{fullpage}
\usepackage{color,xspace}
\usepackage{setspace}

\usepackage{thmtools,thm-restate} 


\newtheorem{theorem}{Theorem}[section]
\newtheorem{lemma}[theorem]{Lemma}

\newtheorem{corollary}[theorem]{Corollary}
\newtheorem{definition}[theorem]{Definition}
\newtheorem{proposition}[theorem]{Proposition}

\newtheorem{claim}[theorem]{Claim}

\newcommand{\MST}{\operatorname{MST}}
\newcommand{\OPT}{\operatorname{OPT}}
\newcommand{\ALG}{\operatorname{ALG}}
\newcommand{\class}{\operatorname{class}}
\newcommand{\rank}{\operatorname{rank}}
\newcommand{\mlast}{\operatorname{\textsc{Mlast-Alg}}}
\newcommand{\type}{\tau}
\newcommand{\fixed}{\sigma}
\newcommand{\incr}{\beta}
\newcommand{\Steiner}{\operatorname{Steiner}}

\newcommand{\load}{\operatorname{load}}
\newcommand{\cost}{\operatorname{cost}}
\newcommand{\str}{3}

\newcommand{\charge}{\phi}
\newcommand{\diam}{\operatorname{diam}}
\newcommand{\net}{Z}
\newcommand{\leaves}{L}

\newcommand{\layer}{H}

\newcommand{\buy}{\mathcal{B}}
\newcommand{\rent}{\mathcal{R}}
\newcommand{\clust}{\mathcal{C}}

\newcommand\E{\mathbb{E}}
\newcommand\R{\mathbb{R}}

\newcommand{\defcal}[1]{\expandafter\newcommand\csname c#1\endcsname{{\mathcal{#1}}}}
\newcommand{\defbb}[1]{\expandafter\newcommand\csname b#1\endcsname{{\mathbb{#1}}}}
\newcounter{calBbCounter}
\forLoop{1}{26}{calBbCounter}{
    \edef\letter{\Alph{calBbCounter}}
    \expandafter\defcal\letter
		\expandafter\defbb\letter
}

\allowdisplaybreaks
\begin{document}

\title{LAST but not Least: Online Spanners for Buy-at-Bulk}

\author{Anupam Gupta\thanks{Computer Science Department, Carnegie Mellon
    University, Pittsburgh, PA 15213, USA. Research partly supported by
    NSF awards CCF-1319811, CCF-1540541 and CCF-1617790.}
\and
R.~Ravi\thanks{Tepper School of Business, Carnegie Mellon University, USA, ravi@cmu.edu; This material is based upon research supported in part by the U. S. Office of Naval Research under award number N00014-12-1-1001, and the U. S. National Science Foundation under award number CCF-1527032.}
\and
Kunal Talwar\thanks{Google Research}
\and
Seeun William Umboh\thanks{Department of Mathematics and Computer
  Science, Eindhoven University of Technology, Netherlands. Email:
  {\tt seeun.umboh@gmail.com}. This work was supported in part by ERC
  consolidator grant 617951 and NSF grant CCF-1320854. Part of this work was done while a student at the University of Wisconsin - Madison, and while visiting the Simons Institute for the Theory of Computing.}
}
\date{\today}

\begin{titlepage}
\def\thepage{}
\thispagestyle{empty}

\maketitle
\begin{abstract}
  The online (uniform) buy-at-bulk network design problem asks us to
  design a network, where the edge-costs exhibit economy-of-scale.
  Previous approaches to this problem used tree-embeddings, giving us
  randomized algorithms. Moreover, the optimal results with a
  logarithmic competitive ratio requires the metric on which the network
  is being built to be known up-front; the competitive ratios then
  depend on the size of this metric (which could be much larger than the
  number of terminals that arrive).

  We consider the buy-at-bulk problem in the least restrictive model
  where the metric is not known in advance, but revealed in parts along
  with the demand points seeking connectivity arriving online. For the
  single sink buy-at-bulk problem, we give a deterministic online
  algorithm with competitive ratio that is logarithmic in $k$, the
  number of terminals that have arrived, matching the lower bound known
  even for the online Steiner tree problem. In the \emph{oblivious} case when the buy-at-bulk function used
  to compute the edge-costs of the network is not known in advance (but
  is the same across all edges), we give a deterministic algorithm with
  competitive ratio polylogarithmic in k, the number of terminals.

  At the heart of our algorithms are optimal constructions for online
  Light Approximate Shortest-path Trees (LASTs) and spanners, and their
  variants. We give constructions that have optimal trade-offs in terms
  of cost and stretch. We also define and give constructions for a new
  notion of LASTs where the set of roots (in addition to the points)
  expands over time. We expect these techniques will find applications in
  other online network-design problems.

\end{abstract}
\end{titlepage}

\setcounter{tocdepth}{1}
\tableofcontents
\newpage

\section{Introduction}
\label{sec:intro}

The model of (uniform) buy-at-bulk network design captures economies-of-scale in
routing problems. Given an undirected graph $G = (V,E)$ with edge lengths $d: E \to
R_{\geq 0}$---we can assume the lengths form a metric---the cost of
sending $x_e$ flow over any edge $e$ is $d(e) \cdot f(x_e)$ where $f$
is some concave cost function. The total cost is the sum over all edges
of the per-edge cost. Given some traffic matrix (a.k.a.\ demand), the
goal is now to find a routing for the demand to minimize the total cost.
This model is well studied both in the operations research and
approximation algorithms communities, both in the offline and online
settings. In the offline setting, an early result was an $O(\log
k)$-approximation due to Awerbuch and Azar, one of the first uses of
tree embeddings in approximation algorithms~\cite{AwerbuchA97}---here
$k$ is the number of demands. 
For the single-sink
case, the first $O(1)$-approximation was given
by~\cite{GuhaMM09}.
In fact, one can get a constant-factor even for the ``oblivious''
single-sink case where the demands are given, but the actual concave
function $f$ is known only after the network is built~\cite{GoelP12}.

The problem is just as interesting in the \emph{online context}: in the
online single-sink problem, new demand locations (called \emph{terminals}) are added over time,
and these must be connected to the central root node as they
arrive. This captures an increasing demand for telecommunication
services as new customers arrive, and must be connected via access
networks to a central core of nodes already provisioned for high
bandwidth. The Awerbuch-Azar approach of embedding $G$ into a tree
metric $T$ with $O(\log n)$ expected stretch (say using~\cite{FRT04}),
and then routing on this tree, gives an $O(\log n)$-competitive
randomized algorithm even in the online case. But \emph{this requires
  that the metric is known in advance}, and the dependence is on $n$, the
number of nodes in the metric, and not on the number of terminals $k$!
This may be undesirable in situations when $n \gg k$; for example, when the terminals come
from a Euclidean space $\R^d$ for some large $d$. Moreover, we only get a \emph{randomized} algorithm
(competitive against oblivious adversaries).\footnote{The tree
  embeddings of Bartal~\cite{Bartal96} can indeed be done online with
  $O(\log k \log \Delta)$ expected stretch, where $\Delta$ is the ratio of maximum to minimum distances in the metric. Essentially, this is because the probabilistic partitions used to construct the embedding can be computed online. This gives an $O(\log^2
  k)$-competitive randomized algorithm, alas sub-optimal by a logarithmic
  factor, and still randomized.}

In this paper, we study the Buy-at-Bulk problem in the online setting,
in the least restrictive model where the metric is not known in advance,
so the distance from some point to the previous points is revealed only
when the point arrives. This forces us to focus on the problem
structure, since we cannot rely on powerful general techniques like tree
embeddings. Moreover, we aim for deterministic algorithms for the
problem. Our first main result is an asymptotically optimal deterministic online
algorithm for single-sink buy-at-bulk.

\begin{restatable}[Deterministic Buy-at-Bulk]{theorem}{babknown}
\label{thm:non-ob-BAB}
    There exists a deterministic $O(\log k)$-competitive algorithm for
    online single-sink buy-at-bulk, where $k$ is the number of
    terminals.
\end{restatable}
Note that the guarantee is best possible, since it matches the lower
bound~\cite{ImaseW91} for the special case of a single cable type
encoding the online Steiner tree problem.

En route, we consider a generalization of the {\em Light
  Approximate Shortest-path Trees} (LASTs). Given a set of ``sources''
and a sink, a LAST is a tree of weight close to the minimum spanning
tree (MST) on the sources and the sink, such that the tree distance from
any source to the sink is close to the shortest-path distance.
%
Khuller, Raghavachari, and Young~\cite{KhullerRY95} defined and studied
LASTs in the offline setting and showed that one can get constant
stretch with cost constant times the MST. Ever since their work, LASTs have proved very
versatile in network design applications. We first give (in
\S\ref{sec:lasts}) a simple construction of LASTs in the online setting
where terminals arrive online. We get constant stretch, and cost at
most $O(\log k)$ times the MST, which is the best possible in the online case.

For our algorithms, we extend the notion of LASTs to the setting of MLASTs
(\emph{Multi-sink LASTs}) where both sources \emph{and sinks} arrive
over time. We have to maintain a set of edges so that every source
preserves its distance to the closest sink arriving before it, at
minimum total cost. We provide a tight deterministic online algorithm
also for MLASTs, which we think is of independent interest. This
construction appears in \S\ref{sec:mlast}. Then we use MLASTs to prove Theorem \ref{thm:non-ob-BAB} in \S\ref{sec:non-obl}.

\paragraph{Oblivious Buy-at-Bulk.} We then change our focus to the
\emph{oblivious} problem. Here we are given neither the terminals nor the
buy-at-bulk function $f$ in advance. When the terminals arrive, we have to
choose paths for them to the root, so that for every concave cost
function $f$, the chosen routes are competitive to the optimal solution
for $f$. Our first result for this problem is the following:

\begin{restatable}[Oblivious Buy-at-Bulk, Randomized]{theorem}{BabUnknownRand}
  \label{thm:ob-BAB-rand}
  There exists a randomized online algorithm for the buy-at-bulk problem
  that produces a routing $P$ such that
for all concave functions $f$,
  \[\textstyle \E\big[\cost_f(P)\big] \leq O(\log^2
  k)\OPT_f.\] 
\end{restatable}
This randomized algorithm has the same approximation guarantee as one
obtained using Bartal's tree-embedding technique. The benefit of this
result, however, is in the ideas behind it. We give constructions of
low-stretch spanners in the online setting. Like LASTs, spanners have
been very widely studied in the offline case; here we show how to
maintain light low-stretch spanners in the online setting (in
\S\ref{sec:spanner}). Then we use the spanners to prove the above theorem in
\S\ref{sec:oblivious-bab-rand}. 
Moreover, building on these ideas, we give a deterministic algorithm in
\S\ref{sec:oblivious-bab-det}.

\begin{restatable}[Oblivious Buy-at-Bulk, Deterministic]{theorem}{BabUnknownDet}
  \label{thm:ob-BAB-det}
  There exists a deterministic online algorithm for the buy-at-bulk
  problem that produces a routing $P$ such that
  for all concave functions $f$,
  \[\cost_f(P) \leq O(\log^{2.5} k)\OPT_f.\]
\end{restatable}

A question that remains open is whether there is an $O(\log
k)$-competitive algorithm for this problem.  The only other
deterministic oblivious algorithm we know for the buy-at-bulk problem is
a derandomization of the oblivious network design algorithm
from~\cite{GuptaHR06}, which requires the metric to be given in advance.

\textbf{LASTs and Spanners.}  A central contribution of our work is to
demonstrate the utility of online spanners in building networks
exploiting economies-of-scale. We record the following two theorems on
maintaining LASTs and spanners in an online setting, since they are of
broader interest. These results are near-optimal, as we discuss in the
respective sections.

\begin{restatable}[Online LAST]{theorem}{thmlast}
  \label{thm:last}
  There exists a deterministic online algorithm for maintaining a
  tree with cost $O(\log k)$ times the MST, and a stretch of $7$ for
  distances from terminals to the sink.
\end{restatable}

\begin{restatable}[Online Spanner]{theorem}{thmspanner}
  \label{thm:spanner}
  There exists a deterministic online algorithm that given $k$ pairs of
  terminals, maintains a forest with cost $O(\log k)$ times the Steiner
  forest on the pairs, and a stretch of $O(\log k)$ for distances
  between all given pairs of terminals. Moreover, the total number of edges 
  in the forest is $O(k)$, i.e. linear in the number of terminals.
\end{restatable}


Our results on online LAST and multicommodity spanners give us an
optimal $O(\log k)$-competitive deterministic algorithms for the
``single cable version'' of the buy-at-bulk problems, where the function
is a single-piece affine concave function~\cite{MansourP98}. 


\subsection{Our Techniques}
\label{sec:techniques}

We now outline some of the ideas behind our algorithms and the role of
spanner constructions in them. All our algorithms share the same
high-level framework: (1) each terminal $v$ is assigned an integer
\emph{type} $\type(v)$; (2) it is then routed through a sequence of
terminals with increasing types; (3) the routes are chosen using a
spanner-type construction. In particular, each algorithm is specified by
a \emph{type-selection rule}, the sequence of terminals to route
through, and whether to use a spanner or MLAST.

The analysis for our deterministic algorithms also has a common thread: while these algorithms are not based on tree-embeddings, the analysis uses tree-embeddings crucially. Both analyses are based on the analysis framework of \cite{Umboh15} and follow the same template: we fix an HST embedding $T$ of the metric, and charge the cost of the algorithm to the cost of the optimal solution on this tree (losing some factor $\alpha$). Since HSTs can approximate general metrics to within a factor of $O(\log k)$, this gives us an $O(\alpha \log k)$-competitive algorithm.

\paragraph{Functions vs.\ Cables.}
As is common with buy-at-bulk algorithms, we represent the function
$f(x)$ as the minimum of affine functions $\min_i \{ \sigma_i + \beta_i
x\}$. And when we route a path, we even specify which of the linear
functions we use on each edge of this path. Each $i$ is called a ``cable
type'', since one can think of putting down cable-$i$ with an up-front
cost of $\sigma_i$, and then paying $\beta_i$ for every unit of flow
over it.

\paragraph{Non-oblivious algorithm.} In the non-oblivious setting where
we know the function $f$,
we will want to route each terminal $v$ through a path $P(v)$ with
non-decreasing cable types. So $\type(v)$ should simply be the cable of
lowest type that we will install on $P(v)$---how should we choose this
value? It makes sense to choose $\type(v)$ based on the number of other
terminals close to $v$. Intuitively, the more terminals that are nearby, the larger the flow that can be aggregated at $v$, making it natural to select a larger type for $v$.

Once we have chosen the type, the route selection is straightforward:
first we route $v$'s demand to the nearest terminal of higher type using
cable type $\type(v)$. Then we iterate: while $v$'s demand is at a
terminal $w$ (which is not the root $r$), we route it to a terminal $w'$ of
type higher than $\type(w)$ that nearest to $w$, using cable type
$\type(w)$.

Finally, how do we select the path when routing $v$'s demand from $w$ to
$w'$? This is where our Multi-sink LAST (MLAST) construction comes in
handy. We want that for each cable type $i$, the set of edges $\layer_i$
on which we install cable type $i$ has a small total
cost 
while ensuring that each terminal of type $i$ has a short path to its
nearest terminal of higher type. We can achieve these properties by
having $\layer_i$ be an MLAST with the sources being the terminals of
type exactly $i$, and the sinks being terminals of type higher than $i$.

\paragraph{Randomized Oblivious Algorithm.}
While designing oblivious algorithms seems like a big challenge because
we have to be simultaneously competitive for all concave cost functions,
Goel and Estrin~\cite{GoelE05} showed that functions of the form $g_i(x)
= \min\{x, 2^i\}$ form a ``basis'' and hence we just have to be good
against all these so-called ``rent-or-buy'' functions. This is precisely
our goal.

Note that the optimum solution for the cost function $g_0$ is the optimal Steiner tree, and that for the cost function $g_{M}$, for $M \gg k$, is the shortest path tree rooted at the sink. Thus being competitive against $g_0$ and $g_M$ already requires us to build a LAST. Thus it is not surprising that our online spanner algorithm is a crucial ingredient in our algorithm.

There are two key ideas. The first is that to approximate a given
rent-or-buy function, it suffices to figure out which terminals should
be connected to the root via ``buy'' edges---i.e., via edges of cost
$2^i$ regardless of the load on them---these terminals we call the
\emph{``buy'' terminals}. The rest of the terminals are simply connected
to the buy terminals via shortest paths. One way to choose a good set of
buy terminals is by random sampling: if we wanted to be competitive
against function $g_i$, we could choose each terminal to be a ``buy''
terminal with probability $2^{-i}$. (See, e.g.,~\cite{AwerbuchAB04}.) Since in the oblivious case we don't
know which function $g_i$ we are facing, we have to hedge our bets.
Hence we choose each terminal $v$ to have $\type(v) = i$ with
probability $2^{-i}$. Thus terminal $v$ is a good buy terminal for all
$g_i$ with $i \leq \type(v)$.

Next, we need to ensure that the path $P(v)$ we choose for $v$
simultaneously approximates the shortest path from $v$ to the set of
terminals with type at least $i$ for all $i$.  This is where our
online spanner construction comes in handy. For each type $i$, we will
build a spanner $F_i$ on terminals of type at least $i$.

\paragraph{Deterministic Oblivious Algorithm.}
To obtain our deterministic oblivious algorithm, we make two further
modifications. First, we remove the need for randomness by using the
deterministic type selection rule of the non-oblivious algorithm. This
modification already yields an $O(\log^3 k)$-competitive algorithm. A technical difficulty that arises is that our type-classes are no longer nested. Thus it is not obvious how to route from a node of type $i$ to one of a higher type, as these nodes may not belong to a common spanner. Adding nodes to multiple spanners can remedy this, but leads to a higher buy cost; being stingy in adding nodes to spanners can lead to higher rent cost. By carefully balancing these two effects and using a more sophisticated routing scheme, we are able to improve the competitive ratio to $O(\log^{2.5} k)$.

\subsection{Other Related Work}
\label{sec:related}

Offline approximation algorithms for the (uniform) buy-at-bulk network
design problem were initiated in~\cite{SalmanCRS97}---here
\emph{uniform} denotes the fact that all edges have the same cost
function $f(\cdot)$, up to scaling by the length of the edge.  Early
results for approximation algorithms for buy-at-bulk network
design~\cite{MansourP98,AwerbuchA97} already observed the relationship
to spanners, and tree embeddings. Using the notion of probabilistic
embedding of arbitrary metrics into tree
metrics~\cite{AKPW95,Bartal96,Bartal98,Bartal04,FRT04}, logarithmic factor approximations
are readily derived for the buy-at-bulk problem in the offline
setting. A hardness result of $O(\log^{\nicefrac14 - \varepsilon} n)$
shows we cannot do much better in the worst case~\cite{Andrews04}.

For the offline single-sink case, new techniques were developed to get
$O(1)$-approx\-imations~\cite{GuhaMM09}, as well as to prove
$O(1)$-integrality gaps for natural LP formulation~\cite{GargKKRSS01,
  Talwar02}; other algorithms have been given
by~\cite{GKR03a,GrandoniI06,JothiR09}. Apart from its inherent interest,
the single-sink buy-at-bulk problem can also be used to solve other
network design problems~\cite{GRS11}. The oblivious single-sink version
was first studied in~\cite{GoelE05}, and $O(1)$-approximations for this
very general version was derived in~\cite{GoelP12}.

In the setting of online algorithms, the Steiner tree and generalized
Steiner forest problems have tight $O(\log k)$-competitive
algorithms~\cite{BermanC97,ImaseW91}, where $k$ is the number of
terminals. These algorithms work in the model where new terminals only
reveal their distances to previous terminals when they arrive, and the
metric is not known \emph{a priori}. It is well-known that the
tree-embedding result of Bartal~\cite{Bartal96} can be implemented
online to give an $O(\log k \cdot \min(\log \Delta, \log
k))$-competitive algorithm for the online single-sink oblivious
buy-at-bulk problem, where $\Delta$ is the ratio of maximum to minimum
distances in the metric. For online rent-or-buy, Awerbuch, Azar, and
Bartal~\cite{AwerbuchAB04} gave an $O(\log^2 k)$-deterministic and an
$O(\log k)$-randomized algorithm; recently,~\cite{Umboh15} gave an $O(\log
k)$-deterministic algorithm.

If the metric is known a priori, the results depend on $n$, the size of
the metric, and not on $k$, the number of terminals.\footnote{Note that $k$ is at most $n$, and it can be much smaller.} E.g.,
tree-embedding results of~\cite{FRT04} give a randomized $O(\log
n)$-competitive algorithm, or a derandomization of oblivious network
design from~\cite{GuptaHR06} gives an $O(\log^2 n)$-competitive
algorithm.


A generalization which we do not study here is the \emph{non-uniform}
buy-at-bulk problem, where we can specify a different concave function
on each edge. A poly-logarithmic approximation for this problem was
recently given by Ene et al.~\cite{EneCKP15}; see the references therein
for the rich body of prior work.

\section{Preliminaries}
\label{sec:prelims}

Formally, in the online buy-at-bulk problem, we have a complete graph $G
= (V,E)$, and edge lengths $d(e)$ satisfying the triangle inequality. In
other words, we can treat $(V,d)$ as a metric. We have $M$ \emph{cable
  types}.  The $i$-th cable type has \emph{fixed cost} $\fixed_i$ and
\emph{incremental cost} $\incr_i$, with $\fixed_i > \fixed_{i-1}$ and
$\incr_i < \incr_{i-1}$. Routing $x$ units of demand through cable type
$i$ on some edge $e$ costs $(\fixed_i + \incr_i x)d(e)$.

In the single-sink version, we are given a root vertex $r \in
V$. Initially, no cables are installed on any edges. When a terminal $v$
arrives, we install some cables on some edges and choose a path $P(v)$
on which to route $v$'s unit demand. (This routing has to be unsplittable,
i.e., along a single path.) We are allowed to install multiple cables on
the same edge; if $e \in P(v)$ has multiple cables installed on it,
$v$'s demand is routed on the one with highest type, i.e., the one with least incremental cost. The choice of
$P(v)$ and cable installations are irrevocable.  Given a routing
solution, if $\load_i(e)$ is the total amount of demand routed through
cable type $i$ on edge $e$, the total cost is $\sum_{e \in E}
\big[\sum_{i : \load_i(e) > 0} \fixed_i + \incr_i \load_i(e) \big]
d(e)$.  We call $\sum_{e \in E} \sum_{i : \load_i(e) > 0} \fixed_i d(e)$
the \emph{fixed cost} of the solution and $\sum_{e \in E} \sum_{i :
  \load_i(e) > 0} \incr_i \load_i(e) d(e)$ the \emph{incremental cost}
of the solution. Let $\OPT$ denote the cost of the optimal solution.  We
assume the cable costs satisfy the conditions of the following
Lemma~\ref{lem:prune}.



\begin{lemma}~\cite{GuhaMM09}
 \label{lem:prune}
 We can prune the set of cables such that for all the retained cable
 types $i$, (a)~$\fixed_i \geq \str \fixed_{i-1}$, and (b)~$\incr_i
 \leq (1/\str^2) \incr_{i-1}$, so that the cost of any solution using
 only the pruned cable types increases only by an $O(1)$ factor.
\end{lemma}






\subsection{HST embeddings}
\label{sec:hst}

Let $(X,d)$ be a metric over a set of points $X$ with distances $d$ at
least $1$. 

\begin{definition}[HST embeddings~\cite{Bartal96}]
  \label{def:HST}
  A \emph{hierarchically separated tree (HST) embedding} $T$ of metric $(X,d)$ is a rooted tree with height $\lceil \log_2 (\max_{u,v \in X} d(u,v))\rceil$ and edge lengths that are powers of $2$ satisfying:
  \begin{enumerate}[topsep=0pt,noitemsep,label=\emph{\alph*}.]
  \item The leaves of $T$ are exactly $X$.
  \item The length of the edge from any node to each of its children is the same.
  \item The edge lengths decrease by a factor of $2$ as one moves along any root-to-leaf path.
  \end{enumerate}
  For $e \in T$, we use $T(e)$ to denote the length of $e$ and say that $e$ is a \emph{level-$j$ edge} if $T(e) = 2^j$. Furthermore, we write $T(u,v)$ to denote the distance between $u$ and $v$ in $T$. We also use $\leaves(e)$ denote the leaves that are ``downstream'' of $e$, i.e.~they are the leaves that are separated from the root of $T$ when $e$ is removed from $T$.
\end{definition}

The crucial property of HST embeddings that we will exploit in our analyses is the following proposition, which follows directly from Properties (2) and (3) of Definition \ref{def:HST}.
\begin{proposition}
  Let $T$ be a HST embedding $T$ of $(X,d)$. For any level-$j$ edge $e \in T$, we have $d(u,v) < 2^j$ for any $u,v \in \leaves(e)$.
\end{proposition}

\begin{corollary}[\cite{AwerbuchA97,FRT04}]
  \label{cor:HST}
  For any online buy-at-bulk instance with $k$ terminals $X$ and distances $d$, there exists a HST embedding $T$ of $(X,d)$ such that $\OPT(T) \leq O(\log k)\OPT$, where $\OPT(T)$ is the cost of the optimal solution for online buy-at-bulk with terminals $X$ in the tree $T$.
\end{corollary}

\subsection{Decomposition into rent-or-buy instances}
\label{sec:rob-decomp}

Since buy-at-bulk functions can be complicated, it will be useful to
deal with simpler \emph{rent-or-buy} functions where to route a load of
$x$ on an edge, we can either \emph{buy} the cable for unlimited use at
its buy cost, or pay the rental cost times the amount $x$. Given a
buy-at-bulk instance as above, for each $i$, define the rent-or-buy
instance with the rent-or-buy function $f_i(x) = \min\{\fixed_i,
\incr_{i-1}x\}$. Let $\OPT_i$ to be the cost of the optimum solution
with respect to this function $f_i$. Note that under this function when
the load aggregates up to $\frac{\fixed_i}{\incr_{i-1}}$, it becomes
advantageous to switch from renting to buying.  


The following lemma will prove very useful, since we can charge
different parts of our cost to different $\OPT_i(T)$s for some HST $T$, and
then sum them up to argue that the total cost is $O(1) \OPT(T)$,
and hence  $O(\log k) \OPT$ using Corollary~\ref{cor:HST}.

\begin{lemma}[$\OPT$ Decomposition on Trees]
  \label{lem:decomp}
  For every tree $T$, we have $\sum_i \OPT_i(T) \leq O(1) \OPT(T)$.
\end{lemma}

\begin{proof}
  For an edge $e \in T$, let $\charge^*_i(e)$ and $\charge^*(e)$ be the costs incurred by $\OPT_i(T)$, and $\OPT(T)$ respectively, on $e$. We will show that for every edge $e \in T$, we have $\sum_i \charge^*_i(e) \leq O(1) \charge^*(e)$. Let $X(e)$ be the set of terminal-pairs whose tree paths (in the single-sink case, terminals whose paths to $r$) in $T$ include $e$, and $T(e)$ denote the length of the tree edge $e$. So, $\charge^*_i(e) = T(e) \cdot \min \{\fixed_i, \incr_{i-1} |X(e)|\}$ and $\charge^*(e) = T(e) \cdot \min_i \{\fixed_i + \incr_i|X(e)|\}$. For any fixed $j$, we have
  \begin{align*}
    \sum_i \min \{\fixed_i, \incr_{i-1} |X(e)|\}
    &\leq \sum_{i \leq j} \fixed_i + \sum_{i > j} \incr_{i-1}|X(e)| \\
    &\leq O(1) (\fixed_j + \incr_j|X(e)|),
  \end{align*}
  where the last inequality follows from the fact that the fixed costs $\fixed_i$ increase geometrically with $i$ and the incremental costs $\incr_i$ decrease geometrically with $i$, as assumed in Lemma~\ref{lem:prune}.

  Thus, $\sum_i \charge^*_i(e) \leq O(1)\charge^*(e)$ for every edge $e \in T$ and so $\sum_i \OPT_i(T) \leq O(1) \OPT(T)$.
\end{proof}

The next lemma proves that the rent-or-buy functions form a ``basis''
for buy-at-bulk functions. For $i\in \mathbb{Z}_{\geq 0}$, define the
rent-or-buy function $g_i(x) = \min\{x, 2^i\}$. (See,
e.g.,~\cite[Section~2]{GoelP12}.) 

\begin{lemma}[Basis of Rent-or-Buy Functions]
  \label{lem:basis}
  Fix some routing of demands. If for every $i$, the cost of this
  routing under $g_i$ is within a factor of $\rho$ of the optimal
  routing for $g_i$, then for any monotone, concave function $f$ with
  $f(0) = 0$, the cost of the routing under $f$ is within $O(\rho)$ of
  the optimal cost of routing under $f$.
\end{lemma}

\section{LASTs}
\label{sec:lasts}

In the \emph{light approximate shortest-path tree} (LAST) problem, we
are given a graph $G = (V,E)$ with edge costs $d(u,v)$ and a root $r$,
and we seek a low-cost spanning subtree $T$ that preserves distances to
the root. Let $\MST$ denote the cost of the minimum spanning tree. We
say that $T$ is an $(\alpha, \beta)$-LAST if $c(T) \leq \alpha\, \MST$ and
$d_T(v,r) \leq \beta\, d(v,r)$ for all $v \in V$. Without loss of
generality, $G$ is complete and $d$ is a metric on $V$. In the online
setting, terminals arrive one by one; when a terminal arrives, its
distances to previous terminals are revealed and the algorithm picks an
edge that connects the new terminal to a previous one. We prove the
following theorem.

\thmlast*



\paragraph{High-level overview.}
Our algorithm maintains a greedy online Steiner tree $T$ plus a set of
edges $A$ which we call the ``direct'' edges. The subgraph $H$ produced
by the graph consists of a subset of edges from $T \cup A$. When
terminal $v$ arrives, the algorithm considers the shortest $(v,r)$-path
$P_v$ contained in $T \cup A$; if $d(P_v) > 7d(v,r)$, then it adds
$(v,r)$ to $H$ and $A$, otherwise it adds $P_v$ to $H$. The formal
algorithm is presented as Algorithm~\ref{alg:last}.

\begin{algorithm}
\caption{Online LAST}
\begin{algorithmic}[1]
  \label{alg:last}
  \STATE $T \leftarrow \emptyset; H \leftarrow \emptyset; A \leftarrow
  \emptyset$
  \WHILE {terminal $v$ arrives}
  \STATE Let $u$ be closest previous terminal
  \STATE Add $(u,v)$ to $T$
  \STATE Let $P_v$ be the shortest $(v,r)$-path in $T \cup A$
  \IF {$d(P_v) > 7d(v,r)$}
  \STATE Add $(v,r)$ to $H$ and $A$
  \ELSE
  \STATE Add $P_v$ to $H$
  \ENDIF
  \ENDWHILE
\end{algorithmic}
\end{algorithm}

\subsection{Analysis}
Since $H \subseteq T \cup A$ and $c(T) \leq O(\log n)\MST$, it
suffices to bound $c(A)$.
\begin{lemma}
  $c(A) \leq 2c(T)$.
\end{lemma}

\begin{proof}
  Define $Z = \{v : (v,r) \in A\}$. Note that
  $c(A) = \sum_{e \in A} d(e) = \sum_{v \in Z} d(v,r)$. For each
  $v \in Z$, define $L_v$ to be the segment of the $(v,r)$-path in $T$
  that is contained in a ball centered at $v$ of radius
  $d(v,r)/2$. The lemma follows if we can show that these segments are
  disjoint. The following claim implies that they are disjoint.
  \begin{claim}
    For every $v,v' \in Z$, we have
    $d_T(v,v') > d(v,r)/2 + d(v',r)/2$.
  \end{claim}

  \begin{proof}
    Suppose, towards a contradiction, that for some $v, v' \in Z$, we
    have $d_T(v,v') \leq d(v,r)/2 + d(v',r)/2$ and $v$ arrived after
    $v'$. Consider the point in time when the $v$ arrived. Since
    $(v',r) \in A$, a possible $(v,r)$-path in $T \cup A$, is to take
    the $(v,v')$-path in $T$ followed by $(v',r)$. So
    \[d(P_v) \leq d_T(v,v') + d(v',r) \leq \frac{d(v,r)}{2} +
    \frac{d(v',r)}{2} + d(v',r).\]
    Furthermore, $v \in Z$ implies $d(P_v) > 7d(v,r)$. Therefore,
    $d(v,r) < d(v',r)/4$.

    On the other hand,
    \begin{align*}
      d(v',r)
      &\leq d(v',v) + d(v,r)\\
      &\leq d_T(v',v) + d(v,r)\\
      &\leq \frac{d(v,r)}{2} + \frac{d(v',r)}{2} + d(v,r)\\
      &< \frac{d(v',r)}{8} + \frac{d(v',r)}{2} + \frac{d(v',r)}{4}\\
      &< d(v',r).
    \end{align*}
    This gives the desired contradiction.
  \end{proof}
  Therefore, the segments $L_v$ are disjoint, so
  $c(A) = 2\sum_{v \in Z} |L_v| \leq 2c(T)$.
\end{proof}


\section{Multi-Sink LASTs}
\label{sec:mlast}

Recall that LASTs were Light Approximate-Shortest-path Trees, i.e.,
trees where we maintain the shortest-path distance of sources to a sink,
using a light tree. (This was traditionally done offline, but in
\S\ref{sec:lasts} we show how to maintain this online---though the
reader is not required to know the offline construction for this section.)
In this section, we are interested in the \emph{multi-sink LAST} (MLAST)
problem.  The input is a sequence of terminals in the underlying metric
$G$, each of which is either a \emph{source} or a \emph{sink} (we assume
the first terminal is always a sink), and our algorithm has to maintain
a subgraph $\layer$ such that (a)~the distance of any source to its
closest sink in $\layer$ should be comparable to the distance of that
source to its closest sink in $G$, and moreover (b)~the cost $d(\layer)$
should not be ``too large''.

Note that when a new sink $s$ arrives, it may be close to many sources,
and hence we may need to add ``shortcut'' edges to reduce their distance
to $s$. Moreover, what does it mean for the cost of $\layer$ to not be
``too large'', since the distance from a source to its closest sink can
fall dramatically over time. Our notion is the following. Note that when
a \emph{source} $v$ arrives, it has to pay at least the distance to its
closest terminal (source or sink) at that time, just to maintain
connectivity. Loosely, we want our cost $d(\layer)$ to be not much more
than the sum of these distances. (N.b.: this is the intuition, formally
we will pay $O(2^{\class(v)})$ which will be defined soon.)

\subsection{The Algorithm}

The first terminal to arrive is a sink we denote as $s^\star$. We use
$R$ and $S$ to denote the sets of sources and sinks that have arrived.
For a terminal $u$, let $R(u)$ and $S(u)$ be the sources and sinks that
arrived strictly before $u$.
Our algorithm $\mlast$ maintains a subgraph $\layer$ that consists of
two parts: 
a forest $F$ which is a ``backbone'' connecting each source to some sink
cheaply, and an edge set $A$ which ``augments'' $F$ to ensure that each
source is not too far from the sinks in $\layer = A \cup F$.

The forest $F$ is constructed using nets which we define now. Let
$\Delta > 0$ be some distance scale. Then, a \emph{$\Delta$-net} $\net$
is a subset of terminals (called \emph{net points}) such that every
terminal $v$ has $d(v,\net) < \Delta$ and for every pair of net points
$u,v \in \net$, we have $d(u,v) \geq \Delta$. Initially, both $F$ and
$A$ are empty. The algorithm maintains a $2^j$-net $\net_j$ on the
entire set of terminals, for each distance scale $j \in
\mathbb{Z}$. When a terminal $v$ arrives, for every distance scale $j$,
it is added to $\net_j$ if $d(u,\net_j) \geq 2^j$. 
The
\emph{class} of $v$ is defined to be $\class(v) := \max \{j : v \in
\net_j\}$, i.e., $\class(v)$ is the largest distance scale such that $v$
belongs to the net of that scale. Hence the class of the first
terminal $s^\star$ is $\infty$. 

Let $u$ be the nearest terminal in
$\bigcup_{j: j > \class(v)} Z_j$, i.e., $u$ is the nearest net-point lying in any net at
a higher scale. 
If the current
vertex $v$ is a source, then the edge $(u,v)$ is added to $F$; otherwise, $F$ is left untouched. Next, the algorithm goes through every source $u$ (including $v$ if it is also a source), checks if $d_{A \cup F}(u,S) \leq \str d(u,S)$, and otherwise adds the edge $(u,u')$ to $A$ where $u' \in S$ is the nearest sink to $u$. 
(See Algorithm \ref{alg:mlast}.)

\begin{algorithm}
\caption{$\mlast(R,S)$}
\begin{algorithmic}[1]
  \label{alg:mlast}
  \STATE $\net_j \leftarrow \emptyset$ for all $j$, $A \leftarrow \emptyset$, $F \leftarrow \emptyset$, $\Delta = 0$;
  \WHILE {terminal $v$ arrives}
  \STATE If $d(v,S(v)) > \Delta$, set diameter $\Delta \leftarrow d(v,S(v))$
  \FOR { all $j \in \mathbb{Z}$}
  \IF {$d(v,\net_j) \geq 2^j$}
  \STATE Add $v$ to $\net_j$
  \ENDIF
  \ENDFOR
  \STATE $\class(v) \leftarrow \max \{j : v \in \net_j\}$
  \IF {$v \in R$} 
  \STATE Add edge $(u,v)$ to $F$ where $u$ is a terminal of higher class nearest to $v$
  \ENDIF
  \FOR {$x \in R$}
  \IF {$d_{A \cup F}(x,S) > \str d(x,S)$}
  \STATE Add $(x,x')$ to $A$ where $x'$ is the terminal in $S$ nearest to $u$
  \ENDIF
  \ENDFOR
  \ENDWHILE
\end{algorithmic}
\end{algorithm}

\textbf{High-Level Idea of the Analysis.}
We first state the main properties of our MLAST algorithm that we will
use in the remainder of the paper.

\begin{restatable}{lemma}{mlastAPI}
  \label{lem:mlast}
  Given an online sequence of sources $R$ and sinks $S$, the algorithm
  $\mlast$ maintains a subgraph $\layer$ and assigns classes to
  terminals such that
  \begin{enumerate}[noitemsep,label=\emph{\alph*}.]
   \item $d(\layer) \leq O(1) \sum_{v \in R} 2^{\class(v)}$,
   \item $d_{\layer}(u,S) \leq \str d(u,S)$ for every $u \in R$,
   \item if $\class(u) = \class(v) = j$, then $d(u,v) \geq 2^j$,
   \item $d(u, R(u) \cup S(u)) \leq 2^{\class(u)} \leq d(u,s^\star)$ for every $u \in R$, where $s^\star$ is the first terminal that arrived, and $R(u)$ and $S(u)$ are the sources and sinks that arrived before $u$.
  \end{enumerate}
\end{restatable}

Property~(a) is the formal bound on the cost of the MLAST. One should
think of $2^{\class(v)}$ as being the radius of some ``dual'' ball
around source $v$, that we will later use to give lower bounds on our
buy-at-bulk instances. Property~(b) guarantees distance preservation.
Properties~(c) and~(d) ensure that terminals of the same class are
well-separated, and that $\class(u)$ is not too large.
Out of these, property~(a) is the most non-trivial one to prove. The
cost of $F$ is easy to bound, since each source $v$ adds in a single
edge of length at most $O(2^{\class(v)})$.  For the cost of edges in
$A$, the argument at a high level is that the sources adding two
different edges of the same length must be far from each other in
$\layer$ (compared to the length of the edges added) else the later one
could use a path to the earlier one, and then the added edge, to fulfill
the distance guarantee.

\subsection{Analysis of $\mlast$}
We devote the remainder of this section to proving Lemma \ref{lem:mlast}. First, note that by construction, $\mlast$
ensures a short path in $H$ from each source $x \in R$ to the sinks $S$:
$d_H(x,S(x)) \leq \str d(x,S(x))$. Secondly, two terminals
$u$ and $v$ having the same class $j$ must both belong to the
$2^j$-net $\net_j$ so
their distance $d(u,v) \geq 2^j$. Thus, we have proved Properties~(b) and~(c).

Next, we prove Property~(d).

\begin{proof}[Proof of Lemma~\ref{lem:mlast}~(d)]
  Let $\class(u) = j$. Since every terminal that arrived before $u$ is at distance $d(u, R(u) \cup S(u))$ from $u$, we have that $u$ belongs to the net $Z_{j^-}$ where $j^- = \lfloor \log d(u, R(u) \cup S(u)) \rfloor$. Thus, $2^j \geq d(u, R(u) \cup S(u))/2$. On the other hand, $s^\star$ belongs to every net since it was the first terminal to arrive. Thus, $u$ cannot belong to any net $Z_{j'}$ where $j' > \lfloor \log d(u,s^\star) \rfloor$.
\end{proof}
 


To prove Property~(a), we start with some useful properties of the
forest $F$ that will allow us to relate $d(\layer)$ to the classes of
sources. Root each tree $T$ in $F$ using the terminal of highest class
in that tree, which we call the \emph{leader} of $T$.

\begin{claim}
  \label{clm:F}
  Each tree $T$ of the forest $F$ satisfies the following properties:
  \begin{enumerate}[noitemsep]
  \item Its leader is a sink and all other terminals are sources.
  \item On every leaf-to-leader path of $T$, the classes of terminals on
    that path strictly increase. 
  \item If $v \in T$ is a source, then when it arrived, the edge $(u,v)$ was added to $F$, where $u$ is a terminal in $T$ with $\class(u) > \class(v)$ and $d(u,v) \leq 2^{\class(v)+1}$.
  \item If $u$ is the parent of $v$ and $v$ is an ancestor of $t$ in $T$, then the length of the $(t,u)$-path $P$ in $T$ is $d_T(t,u) \leq 2^{\class(v)+2} \leq 2^{\class(u)+1}$.
  \end{enumerate}
  See Figure \ref{fig:F} for an illustration of properties (1), (2), and (3); and Figure \ref{fig:F2} for property (4).
  \end{claim}

\begin{proof}
  \begin{figure}[t]
    \centering
    \includegraphics[page=1, scale=0.5]{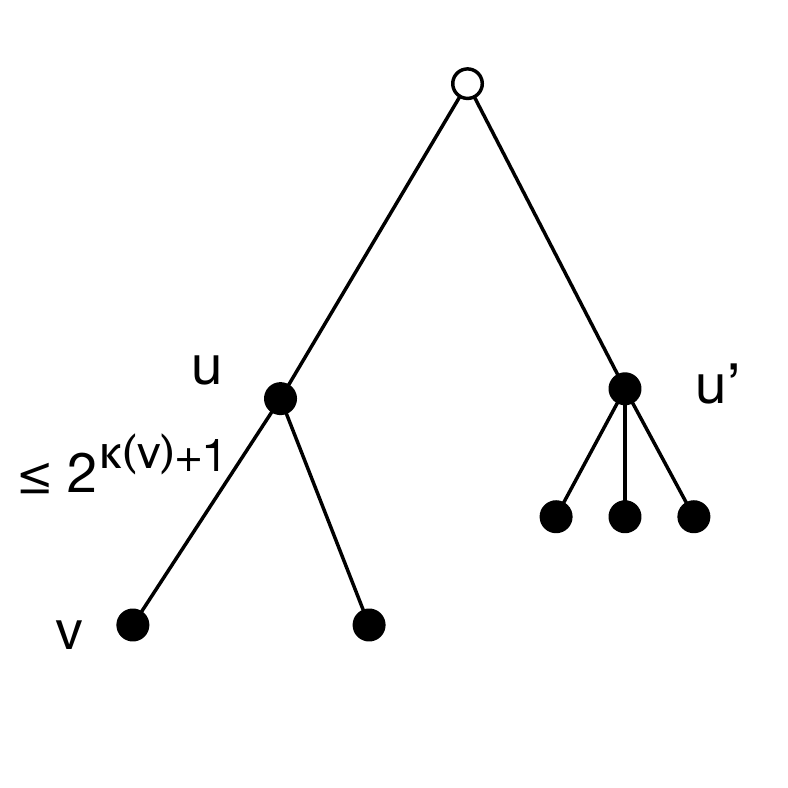}
    \caption{\em An example tree of F. The hollow vertex at the top is
      the leader and a sink, and the other are terminals. When $v$
      arrived, the algorithm added the edge $(u,v)$, making $v$ a child
      of $u$, and the length of the edge is at most
      $2^{\class(v)+1}$. Note that it is possible for two terminals $u$
      and $u'$ to have the same class, and yet $u$'s children have much
      larger class than $u'$'s.}
    \label{fig:F}
  \end{figure}

  Properties (1) and (2) follow from the fact that each source has a
  single edge in $T$ to a terminal of higher class, and sinks are not
  connected to other terminals when they arrive.

  We now prove property (3). Let $(u,v)$ be the edge added when $v$
  arrived. The algorithm chose $u$ to be the nearest terminal of some
  higher class. By definition, all terminals in $\bigcup_{j: j >
    \class(v)} \net_{j}$ have class at least $\class(v)+1$. Moreover,
  the closest net-point to $v$ in $\net_{\class(v)+1}$ is within
  distance $2^{\class(v)+1}$, so $u$ can be no further. Hence, we have
  $d(v,u) \leq d(v, \net_{\class(v)+1}) \leq 2^{\class(v)+1}$.

    \begin{figure}[t]
      \centering
      \includegraphics[page=2, scale=0.5]{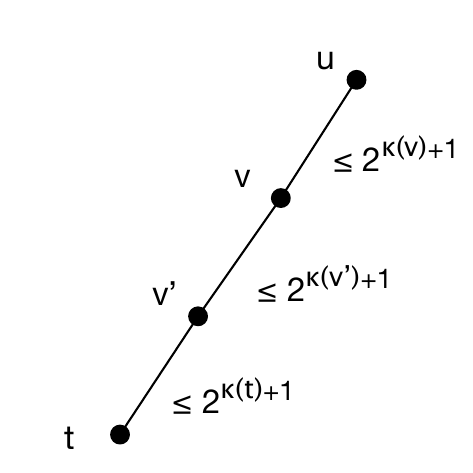}
      \caption{\em An illustration of the proof of property (4). Since $\class(t) < \class(v') < \class(v) < \class(u)$ by properties (2) and (3), the length of the path from $t$ to $u$ is at most $2^{\class(v)+2} \leq 2^{\class(u)+1}$.}
      \label{fig:F2}
    \end{figure}

    Finally, we prove property (4).  Suppose $u$ is the parent of $v$
    and $v$ is an ancestor of $t$ in the tree $T$, and let $P$ be the
    path in $T$ from $t$ to $u$. Property (2) implies that the classes
    of terminals in $P$ are strictly increasing as we move from $t$ to
    $u$. Property (3) now implies that the length of the $t$-$u$ path
    $P$ is at most $\sum_{x \in P: x \neq u} 2^{\class(x)+1} <
    2^{\class(v)+2} \leq 2^{\class(u)+1}$. This concludes the proof of
    the claim.
  \end{proof}

Since each edge in $F$ is added by a distinct source, the total cost
of $F$ (using Claim~\ref{clm:F}(3)) is:
\begin{equation}
  \label{eq:Fcost}
  d(F) \leq 2\sum_{v \in R} 2^{\class(v)}.
\end{equation}
It remains to bound $d(A)$ in terms of the classes of sources. The following claim allows us to find sources to charge $d(A)$ to. Define $A_j = \{(u,v) \in A : d(u,v) \in [2^j, 2^{j+1})\}$. Observe that every edge $(u,v) \in A$ has one endpoint being a source terminal and another being a sink terminal.

\begin{claim}
  \label{clm:netpath}
  Suppose $(u,v) \in A_j$ with $u$ being the source endpoint and $T$ is
  the tree in $F$ containing $u$. Then there exists an ancestor
  $\varphi(u) = z$ of $u$ in $T$ such that $z$ is a source, has class $\class(z) \geq j-2$, and the length of the $(u,z)$-path in $T$ is $d_T(u,z) \leq 2^{j-1}$.
\end{claim}



\def\leader{a}
\begin{proof}
  Let $\leader$ be the sink that is the leader of $T$. By definition of the
  algorithm, $v$ is the nearest sink to $u$ at the moment that $(u,v)$
  was added to $A$, so $d(u,v) \leq d(u,\leader) \leq d_T(u,\leader) \leq
  2^{\class(\leader)+1}$, where the last inequality follows from Claim
  \ref{clm:F}(4). Since $(u,v)$ belongs to $A_j$, we also have $d(u,v)
  \geq 2^j$. Putting the previous two inequalities together, $2^j \leq
  2^{\class(\leader)+1}$ and so the leader $\leader$ has $\class(\leader) \geq j-1$.

  If $u$ has $\class(u) \geq j-2$ we are done, so assume $\class(u) \leq
  j-3$. Let $t$ be the highest ancestor of $u$ such that $\class(t) \leq
  j-3$ (clearly $u$ is a candidate). Observe that $t \neq \leader$. Let $z$ be
  $t$'s parent; we claim $z$ satisfies the desired conditions. It has
  class $\geq j-2$ (by the maximality of $t$), and $d(u, z) \leq
  2^{\class(t)+2} \leq 2^{j-1}$ by Claim~\ref{clm:F}(4).  But we'd
  already observed that $d(u,\leader) \geq 2^j$, so $z \neq \leader$; hence $z$ is a
  source by Claim~\ref{clm:F}(1).
\end{proof}

\begin{proof}[Proof of Lemma~\ref{lem:mlast}~(a)]
  To upper bound $d(A)$, we charge the length of each edge $(u,v) \in A$
  to a source vertex, and then show that the charge received by each
  source $z$ is at most $O(1)2^{\class(z)}$. For each edge $(u,v) \in
  A$, the charging scheme is defined as follows: if $(u,v) \in A_j$
  (i.e. its length $d(u,v) \in [2^j, 2^{j+1})$) and $u$ is the source
  endpoint and contained in the tree $T$ of $F$, charge $2^{j+1} \geq
  d(u,v)$ to the ancestor $\varphi(u)$ given by Claim
  \ref{clm:netpath}. Hence $d(A)$ is at most the total charge received
  by all the sources.

  Next, we show that the charge received by every source $z$ is at most
  $O(1) 2^{\class(z)}$. Terminal $z$ is only charged by edges belonging
  to $A_j$ for $j \leq \class(z)+2$; each such edge charges $z$ charges
  $2^{j+1}$. We claim that for each $j$, at most one edge in $A_j$
  charges $z$: this means the total charge received by $z$ is at most
  $\sum_{j \leq \class(z)+2} 2^{j+1} \leq O(1) 2^{\class(z)}$, so $d(A)
  \leq \sum_{z \in R} O(1) 2^{\class(z)}$, which proves the lemma.

  To prove the claim, suppose that there exist two edges $(u,v), (u',v')
  \in A_j$ that charge $z$. Let $u$ and $u'$ be sources, and $v$ and $v'$ be
  sinks. Since both edges charge $z$, $\varphi(u) = \varphi(u') = z$.
Since both edges belong
  to $A_j$, their lengths must satisfy
  \begin{equation}
    \label{eq:uv}
    2^j \leq d(u,v) \leq 2d(u',v').
  \end{equation}
  By definition of the charging scheme, 
  \begin{equation}
    \label{eq:u'z}
    d_T(u',z) \leq 2^{j-1} \qquad \text{and} \qquad d_T(u,z) \leq 2^{j-1}.
  \end{equation}
  Suppose $(u',v')$ was added to $A$ after $(u,v)$. At the time when we
  consider adding $(u',v')$, the sink $v \in S(u')$ and so
  $d_{\layer}(u',S) \leq d_{\layer}(u',v)$. Moreover, there is a path in
  $\layer$ from $u'$ to $v$: the path in $T$ from $u'$ to $u$, followed
  by the edge $(u,v)$, which implies
  \[d_{\layer}(u',S) \leq d_{T}(u',z) + d_{T}(z,u) + d(u,v).\] (See
  Figure \ref{fig:A} for an illustration.) Plugging in the inequalities
  \eqref{eq:uv} and \eqref{eq:u'z}, we get that $d_{\layer}(u',S) \leq
  2^j + 2d(u',v') \leq 3d(u',v')$, where the last inequality follows from the fact that $2^j \leq d(u',v')$. 
  However, this contradicts the condition for adding $(u',v')$ to
  $A$. Hence only one edge in $A_j$ can charge any $z$; this proves the
  claim and hence the lemma.
\end{proof}

    \begin{figure}[t]
      \centering
      \includegraphics[page=3, scale=0.5]{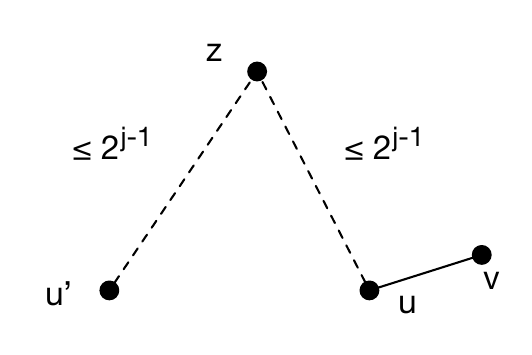}
      \caption{\em An illustration of the situation when two edges $(u,v)$ and $(u',v')$ in $A_j$ charge the same terminal $z$, and $u$ and $u'$ are the source endpoints. The dashed edges represent paths from $u'$ to $z$ and $u$ to $z$ in $T$. There is a path in $H$ from $u'$ to $S$ following the paths in $T$ represented by the dashed edges and the edge $(u,v)$.}
      \label{fig:A}
    \end{figure}

\section{Non-Oblivious Buy at Bulk}
\label{sec:non-obl}

In this section we prove our main result for online buy-at-bulk for the
case when the function $f$ is known. 
\babknown*

For a high-level intuition behind the algorithm, see the discussion in
\S\ref{sec:techniques}. Here we explain precisely how to assign types,
update layers and route demands.

\paragraph{Type assignment.}
Consider a terminal $v$. At the time of its arrival, for each type
$i$, let $d_i(v)$ be the distance to the nearest type-$i$ terminal,
and define the ball $B_i(v) := \{u : d(u,v) \leq d_i(v)/2^3\}$. Let
$n_i(v) := |B_i(v)|$ be the number of terminals in this ball. Terminal
$v$ is assigned type
\[ \textstyle \type(v) := \max \big\{i : n_i(v) \geq \frac{\fixed_i}{\incr_{i-1}}\big\}.\]
To make sense of this, observe that the threshold for being of a certain
type is the same as the threshold for the load at which it is
advantageous to buy rather than rent for that type. This means vertices
assigned a certain type can collect enough ``witnesses'' whose paths to
the root in the optimal solution can be used to pay for edges
constructed by $v$ in the online solution.  

\paragraph{Layers.} 
Let $R_i $ be the current set of terminals of type exactly $i$, and
$X_i := \cup_{i' \geq i} R_{i'}$ be the current set of terminals of
type at least $i$.  Let $X_{i}(v)$ (resp., $R_i(v)$) be the set of
terminals of type at least (resp., exactly) $i$ that arrived strictly
before $v$. Layer $\layer_i$ is maintained by running the online
algorithm $\mlast$ with $R_i$ as sources and $X_{i+1}$ as sinks. This
ensures that every $u \in R_i$ has
$d_{\layer_i}(u,X_{i+1}(v)) \leq \str d(u,X_{i+1}(v))$.

\paragraph{Routing.}
We route $v$'s demand to the root on path $P(v)$ that is constructed
as follows. The path $P(v)$ is constructed iteratively and consists of
$M$ different segments $P_i(v)$, one per type. Initially, $v$'s demand
is located at $v$ and the segments are all empty. While the current
terminal $w$ containing $v$'s demand is not the root, choose $w'$ to
be the terminal in $\layer_{\type(w)}$ with $\type(w') > \type(w)$
nearest to $w$ (nearest in terms of distances in $\layer_{\type(w)}$)
and route $v$'s demand to $w'$ along the shortest path between $w$ and
$w'$ in $\layer_{\type(w)}$.

This completes the description of the algorithm. See Algorithm
\ref{alg:main} for a summary.

\begin{algorithm}
\caption{Online Algorithm for Single-Sink Buy-at-Bulk}
\begin{algorithmic}[1]
\label{alg:main}
  \STATE $\type(r) \leftarrow \infty$
  \STATE $\layer_i \leftarrow \emptyset$ for all $i$
  \WHILE {terminal $v$ arrives}
  \STATE \textit{// Determine type of $v$}
  \STATE Let $d_i(v) = d(v,X_i)$, $B_i(v) = \{u : d(u,v) \leq d_i(v)/2^3\}$, and $n_i(v) = |B_i(v)|$ for $1 \leq i \leq M$
  \STATE $\type(v) \leftarrow \max \{i : n_i(v) \geq \frac{\fixed_i}{\incr_{i-1}}\}$
  \STATE \textit{// Update layers $\layer_i$}
  \FOR {$i \leq \type(v)$}
  \STATE Update $\layer_i$ using $\mlast$ with $R_i(v)$ as sources and $X_{i+1}(v)$ as sinks
  \STATE Install cable of type $i$ on new edges of $H_i$
  \ENDFOR
  \STATE \textit{// Find routing path $P(v)$ to route $v$'s demand to the root}
  \STATE $P_i(v) \leftarrow \emptyset$ for all $i$; $P(v) \leftarrow \emptyset$
  \WHILE {$v$'s demand is not yet at $r$}
  \STATE Let $w$ be the terminal containing $v$'s demand and $i = \type(w)$
  \STATE Let $P_i(v) \subseteq \layer_i$ be the shortest path from $w$ to $X_{i+1}(v)$ in $\layer_i$
  \STATE Route $v$'s demand from $w$ to $X_{i+1}(v)$ along the path $P_i(v)$
  \ENDWHILE
  \ENDWHILE
\end{algorithmic}
\end{algorithm}

\subsection{Analysis}
\label{sec:analysis}







Let us now consider the cost of the algorithm's solution due to each
cable type. The fixed cost due to cables of type $i$ is
$\fixed_i d(\layer_i)$ since they are installed on layer
$\layer_i$. For each terminal $v$ of type $\type(v) \leq i$, $P_i(v)$
is the segment of its routing path consisting of type $i$ cables.
Hence the total cost
of the algorithm's solution is
\begin{gather}
  \ALG := \sum_i \bigg(\fixed_i d(\layer_i) + \sum_{v : \type(v) \leq i}
  \incr_i d(P_i(v))\bigg). \label{eq:alg-cost}
\end{gather}

\paragraph{Proof outline.}
We want to show that $\ALG \leq \OPT(T)$ for any HST embedding $T$,
since Corollary~\ref{cor:HST} would then imply that $\ALG \leq O(\log
k)\OPT$. The key is to decompose $\OPT$ on trees as follows: in
\S\ref{sec:rob-decomp} we defined the rent-or-buy functions $f_i(x) :=
\min\{\fixed_i, \incr_{i-1}\}$, and defined $\OPT_i$ to be the cost of
the optimal solution with respect to $f_i$. Our proof (specifically
Lemmas \ref{lem:fixed-cost} and \ref{lem:rent-cost}) will use the
structure of the HST to give a charging argument showing that
\begin{equation}
\fixed_i d(\layer_i) + \sum_{v: \type(v) < i} \incr_{i-1} d_i(v) \leq
O(1)\OPT_i(T).\label{eq:decomp}
\end{equation}
Roughly speaking, we argue that $d(\layer_i)$ can be charged to the
optimal cost of the witnesses of the terminals of type exactly $i$, and
moreover, the type selection rule forces terminals of type less than $i$
to be spread out.

Thus, if we could prove that $d(P_i(v)) \leq O(1) d_{i+1}(v)$, then we
would be done, because we could decompose the expression for $\ALG$
in~(\ref{eq:alg-cost}) as follows:
\begin{gather}
  \label{eq:ALG-decomp}
  \ALG \leq O(1) \cdot \sum_i \bigg(\fixed_i d(\layer_i) + \sum_{v:
      \type(v) < i} \incr_{i-1} d_i(v)\bigg).
\end{gather}
Inequality \eqref{eq:decomp} and Lemma \ref{lem:decomp} would imply
that $\ALG \leq O(1) \OPT(T)$ for any HST embedding $T$, as desired.

However, $d(P_i(v))$ can be much larger than $d_{i+1}(v)$. This is
because of the ``selfishness'' of the routing scheme: when a terminal
$w$ receives $v$'s demand, it simply routes it to the terminal $w'$ of
higher type nearest to $w$, without any regard as to how far $w'$ is
from $v$. Fortunately, the fact that the incremental costs $\incr_i$
are geometrically decreasing allows us to show that the total
incremental cost incurred over all segments $P_i(v)$ are bounded. In
particular, we show that
$\sum_{v : \type(v) \leq i} \incr_i d(P_i(v)) \leq O(1) \sum_i \sum_{v
  : \type(v) \leq i} \incr_i d_{i+1}(v)$.
The rest of the proof will follow this general outline.



 

\paragraph{Bounding incremental cost.}
We begin by proving the above bound on the incremental cost.
\begin{lemma}
  \label{lem:incr-bound}
  The incremental cost $\sum_i \sum_{v : \type(v) \leq i} \incr_i
  d(P_i(v)) \leq O(1) \sum_i \sum_{v : \type(v) \leq i} \incr_i
  d_{i+1}(v)$. 
\end{lemma}

\begin{proof}
  For terminal $v$, suppose its type-$i$ segment $P_i(v)$ is non-empty: 
let $w_i$ be the starting terminal of the segment. Since $P_i(v)$ starts
with a terminal of type $i$ and ends with a terminal of type at least
$i+1$, so $\type(w_i) = i$. 
Now since $P_i(v)$ is a shortest path in $\layer_i$ from $w_i$ to
$X_{i+1}(v)$, Lemma~\ref{lem:mlast}(b) gives us $d(P_i(v)) \leq \str
\cdot d(w_i, X_{i+1}(v))$. Since $d_{i+1}(v)$ is the distance
from $v$ to $X_{i+1}(v)$, $d(w_i, X_{i+1}(v)) \leq d(w_i, v) + d_{i+1}(v)$.
Moreover, the segments $P_{i'}(v)$ for $i' < i$ form a path from $v$ to $w_i$, so
$d(w_i,v) \leq \sum_{i' < i} d(P_{i'}(v))$. Combining these
inequalities, we get $d(P_i(v)) \leq \str \left(d_{i+1}(v) + \sum_{i' <
    i} d(P_{i'}(v))  \right)$. Unrolling the recursion, this
gives
$d(P_i(v)) \leq \sum_{i' \leq i} \str^{i-i'+1}d_{i'+1}(v)$.
Hence, we get the following bound on the incremental cost of $v$:
  \begin{align*}
    \sum_{i \geq \type(v)} \incr_i \cdot d(P_i(v))
    &\leq \sum_{i \geq \type(v)} \incr_i \cdot \sum_{i' \leq i} \str^{i-i'+1}d_{i'+1}(v)
    = \sum_{i \geq \type(v)} d_{i+1}(v) \cdot \sum_{i' \geq i} \str^{i'-i+1}\incr_{i'},
  \end{align*}
  where the equality follows from rearranging the sums. But the
  $\beta_i$s decrease geometrically by a factor of $3^2$, so the last
  sum is dominated by the first term, and is $O(\beta_i)$.  Hence the proof.
\end{proof}

Lemma \ref{lem:incr-bound} implies Inequality \eqref{eq:ALG-decomp}. 
Define $\ALG_i = \fixed_i d(\layer_i) + \sum_{v: \type(v) < i}
\incr_{i-1} d_i(v)$. 
The rest of this section will show that on any HST embedding $T$,
$\ALG_i \leq O(1) \OPT_i(T)$ for every $i$. Lemma \ref{lem:decomp}
then implies that $\ALG \leq O(1)\OPT(T)$ for any HST embedding $T$
and so $\ALG \leq O(\log k)\OPT$ by Corollary~\ref{cor:HST}.

\paragraph{Charging to HST embeddings.}
For the following, fix an HST embedding $T$. Recall that for edge $e
\in T$, $L(e)$ denotes the leaves below $e$, $T(e)$ the length of $e$,
and $T(u,v)$ the distance between $u$ and $v$ in $T$. Also, observe that
an edge $e$ such that $r \notin L(e)$, the terminals in $L(e)$ either
have to buy the edge at cost $\fixed_i T(e)$, or rent it at cost
$\incr_{i-1}|L(e)| T(e)$. We record this lower bound for later.
\begin{equation}
  \label{eq:ROB-lb}
  \OPT_i(T) \geq \sum_{e \in T : r \notin L(e)} T(e)\min \{\fixed_i, \incr_{i-1}|L(e)|\}.
\end{equation}
To upper-bound $\ALG_i$, we bound both $\fixed_i d(\layer_i)$ and
$\sum_{v: \type(v) < i} \incr_{i-1} d_i(v)$ separately by
$\OPT_i(T)$. In each case, we proceed by developing an appropriate
charging scheme that charges to the edges of $T$ and then arguing that
the total charge received by each edge $e \in T$ is at most a constant
times its contribution to the lower bound of $\OPT_i(T)$ in Inequality
\ref{eq:ROB-lb}.


First, we bound $\sum_{v: \type(v) < i} \incr_{i-1} d_i(v)$. The
charging scheme is as follows: for each terminal $v$, charge
$\incr_{i-1}$ to an edge in $T$ whose length is proportional to
$d_i(v)$ and which contains $v$ as a leaf. Then we argue that no edge
is overcharged; in particular, for every edge $e \in T$, the total
number of terminals that can charge $e$ is at most
$\fixed_i/\incr_{i-1}$. Finally, we use the fact that terminals
charging to $e$ are close together (by the bounded diameter property
of HSTs), so if there were more than $\fixed_i/\incr_{i-1}$ terminals
charging $e$, then the one that arrived last would have been assigned
a type of at least $i$.

\begin{lemma}
  \label{lem:rent-cost}
  $\sum_{v: \type(v) < i} \incr_{i-1} d_i(v) \leq O(1)\OPT_i(T)$.
\end{lemma}

\begin{proof}
  Consider the following charging scheme. For each terminal $v$ with
  type $\type(v) < i$, if $d_i(v) \in [2^j, 2^{j+1})$, charge
  $2^{j+1}\incr_{i-1}$ to the length $2^{j-4}$ edge $e \in T$ whose
  leaves $L(e)$ contain $v$ but not $r$. Such an edge must exist since
  otherwise $d_i(v) \leq d(v,r) < 2^j$.
  The total charge received by the edges of $T$ is at least $\sum_{v :
    \type(v) < i} \incr_{i-1} d_i(v)$, and only edges $e$ with $r
  \notin L(e)$ were charged.

  Consider an edge $e \in T$ of length $2^{j-4}$. Let $C(e) \subseteq
  L(e)$ be the set of terminals charging $e$. We claim that $|C(e)|
  \leq \fixed_i/\incr_{i-1}$. Note that the total charge received by
  $e$ is \[2^{j+1}\incr_{i-1}|C(e)| ~~=~~ 2^5\, T(e) \incr_{i-1}|C(e)| ~~\leq~~
  2^5\, T(e) \min\{\fixed_i, \incr_{i-1}|L(e)|\}.\]
  By \eqref{eq:ROB-lb}, this proves the lemma.
  
  Now to prove the claim. Suppose for a contradiction, that $|C(e)| >
  \fixed_i/\incr_{i-1}$. Let $v$ be the last-arriving terminal of
  $C(e)$. Since $e$ is a length $2^{j-4}$ edge and $v$ charged $e$, we
  have $d_i(v) \geq 2^j$ and $\type(v) < i$. Moreover, since $C(e)
  \subseteq L(e)$, we have that $\diam(C(e)) < 2^{j-3} \leq
  d_i(v)/2^3$. Thus, every terminal $u \in C(e)$ is within distance
  $d_i(v)/2^3$ from $v$ so $C(e) \subseteq B_i(v)$. This implies that
  $|B_i(v)| \geq |C(e)| \geq \fixed_i / \incr_{i-1}$ so $v$ should
  have been assigned a type that is at least $i$, contradicting the
  fact that $\type(v) < i$. Therefore, $|C(e)| \leq
  \fixed_i/\incr_{i-1}$, as desired.
\end{proof}

\paragraph{Bounding the Fixed Cost of $\layer_i$.} There
are three steps to the proof. The first step is to use Lemma
\ref{lem:mlast} to charge the fixed cost to the terminals of type
$i$. The second is to argue that since each terminal $v$ of type $i$
has at least $\fixed_i/\incr_{i-1}$ witnesses nearby, $v$ can use the
cost incurred by $\OPT_i(T)$ in routing its witnesses to pay off the
charge it accumulated. Finally, we use the fact that terminals of type
$i$ that accumulate similar charges must be far apart to argue that no
witness is overcharged.
  
  \begin{lemma}
    \label{lem:fixed-cost}
    $\fixed_i d(\layer_i) \leq O(1)\OPT_i(T)$.
  \end{lemma}

  \begin{proof}
    The layer $\layer_i$ is an MLAST whose set of sources is $R_i$ (the
    terminals of type exactly $i$) and sinks $S_i = X_{i+1}$, the
    terminals of higher type. Let $\class_i(v)$ be the class assigned by
    this MLAST algorithm to terminal $v$. By Lemma \ref{lem:mlast}, we
    have $\fixed_i d(\layer_i) \leq O(1) \sum_{v \in R_i} \fixed_i
    2^{\class_i(v)}$. Define $R_i(j) = \{v \in R_i : \class_i(v) =
    j\}$. 
    Let $E_j$ be the set of length $2^{j-4}$ edges of $T$. To prove the lemma, we will show that for each $j$,
    \begin{equation}
      \label{eq:fixed-cost}
    \fixed_i|R_i(j)| \leq \sum_{e \in E_j : r \notin L(e)} \min \{\fixed_i, \incr_{i-1}|L(e)|\}.
  \end{equation}
  This would then imply that
  \[\fixed_i d(H_i) \leq O(1) \sum_{v \in R_i} \fixed_i
  2^{\class_i(v)} \leq O(1)\sum_j 2^j \fixed_i|R_i(j)| \leq O(1) \sum_j
  2^{j-4} \sum_{e \in E_j : r \notin L(e)} \min \{\fixed_i,
  \incr_{i-1}|L(e)|\}.\] But $T(e) = 2^{j-4}$ for edges $e \in E_j$, so
  \eqref{eq:ROB-lb} bounds the cost by $O(1)\OPT_i(T)$ to complete the
  proof.

\medskip
  Now to prove~(\ref{eq:fixed-cost}).  We will show
  that for each $v \in R_i(j)$ (i.e., having type $i$, and class $j$ in the
  MLAST $H_i$ corresponding to type-$i$ terminals), every terminal $u$ in
  its ball $B_i(v)$ has to be routed on some level-$j$ edge $e \in E_j$
  with $r \notin L(e)$, and that the level-$j$ edges used by $B_i(v)$ is
  disjoint from the edges used by $B_i(v')$ for any other $v' \in
  R_i(j)$. More formally, for each terminal $v$, define $e_j(v)$ to be
  the unique edge of $E_j$ such that $v \in L(e_j(v))$ and define
  $E_j(v) := \{e_j(u) : u \in B_i(v)\}$. We need the following claims.

  \begin{claim}
    \label{clm:ball}
    For every $v \in R_i(j)$ and $u \in B_i(v)$, we have $d(u,v) \leq 2^{j-3}$.
  \end{claim}

  \begin{proof}
    By definition of $B_i(v)$, we have the following bound: $d(u,v) \leq
    d_i(v)/2^3 = d(v,X_i(v))/2^3$. Now observe that $X_i(v) = R_i(v)
    \cup S_i(v)$, where $R_i(v)$ and $S_i(v)$ are the sources and sinks
    that arrived before $v$ in the MLAST for type $i$. By Lemma
    \ref{lem:mlast}(d), $d(v, X_i(v)) \leq 2^{\class_i(v)} = 2^j$. Combining
    this with the above bound on $d(u,v)$ gives us $d(u,v) \leq
    2^{j-3}$, as desired.
    \end{proof}

    \begin{claim}
      \label{clm:disjointness}
      For every $v \in R_i(j)$, we have $\bigcup_{e \in E_j(v) : r \notin L(e)} L(e) \supseteq B_i(v)$. Moreover, $E_j(v) \cap E_j(v') = \emptyset$ for distinct $v,v' \in R_i(j)$.
    \end{claim}

    \begin{proof}
      To prove the first part of the claim, observe that by definition of $E_j(v)$, we have $\bigcup_{e \in E_j(v)} L(e) \supseteq B_i(v)$. Thus, it suffices to show that for every $v \in R_i(j)$ and each terminal $u \in B_i(v)$, the edge $e_j(v)$ does not contain $r$ as a leaf. Suppose, towards a contradiction, that there exists $u \in B_i(v)$ such that $r \in L(e_j(u))$. Since $e_j(u)$ has length $2^{j-4}$, we have $\diam(L(e_j(u))\leq 2^{j-3}$; moreover, $r, u \in L(e_j(u))$ and so $d(u,r) \leq 2^{j-3}$. Now, Claim \ref{clm:ball} implies that $d(u,v) \leq 2^{j-3}$. Therefore, $d(v,r) \leq d(u,v) + d(u,r) < 2^j$. However, this contradicts Lemma \ref{lem:mlast}(d), which implies that $d(v,r) \geq 2^j$. Thus, $r \notin L(e_j(u))$.

      To prove the second part of the claim, suppose towards a contradiction, that there exist $v,v' \in R_i(j)$ and $u \in B_i(v)$ and $u' \in B_i(v')$ such that $e_j(u) = e_j(u')$. By triangle inequality, $d(v,v') \leq d(u,v) + d(u,u') + d(u',v')$. Claim \ref{clm:ball}  implies that $d(u,v), d(u',v') \leq 2^{j-3}$. Since $u,u'$ are leaves of the same edge of length $2^{j-4}$, we get $d(u,u') \leq 2^{j-3}$. Therefore, $d(v,v') < 2^j$. On the other hand, Lemma \ref{lem:mlast}(c) says that $d(v,v') \geq 2^j$. This gives us our desired contradiction.
    \end{proof}

    With these claims in hand, we have
    \begin{align*}
      \sum_{e \in E_j : r \notin L(e)} \min \{\fixed_i, \incr_{i-1}|L(e)|\}
      &\geq \sum_{v \in R_i(j)} \sum_{e \in E_j(v) : r \notin L(e)}\min \{\fixed_i, \incr_{i-1}|L(e)|\}\\
      &\geq \sum_{v \in R_i(j)} \min \{\fixed_i, \incr_{i-1}\sum_{e \in E_j(v) : r \notin L(e)}|L(e)|\}\\
      &\geq \sum_{v \in R_i(j)} \min \{\fixed_i, \incr_{i-1}|B_i(v)|\} \\
      &\geq \sum_{v \in R_i(j)} \fixed_i = \fixed_i |R_i(j)|,
    \end{align*}
    where the first inequality follows from the second part of Claim \ref{clm:disjointness}, the third inequality from the first part, and the last inequality from the fact that $|B_i(v)| \geq \fixed_i/\incr_{i-1}$.
  \end{proof}

  Having proved these two lemmas, we have that for every HST embedding
  $T$ and every $i$, we have $\ALG_i \leq O(1)\OPT_i(T)$. Summing over
  all $i$ and using Lemma \ref{lem:decomp}, we have $\ALG \leq
  O(1)\OPT(T)$ for any HST embedding $T$, and so $\ALG \leq O(\log
  k)\OPT$. This proves Theorem \ref{thm:non-ob-BAB}.


\section{Multi-Commodity Spanners}
\label{sec:spanner}

In the \emph{online $\alpha$-MCS} problem, we are given a graph $G =
(V,E)$ with edge costs $d(u,v)$ satisfying the triangle
inequality. Terminal pairs $(s,t)$ arrive one-by-one. The goal is to
maintain a minimum-cost subgraph $H$ such that $d_H(s,t) \leq \alpha
d(s,t)$ for all terminal pairs $(s,t)$. We say that $H$ is
$\beta$-competitive if its cost $d(H) = \sum_{(u,v) \in H} d(u,v)$ is at most $\beta$ of the
cheapest Steiner forest connecting the terminal pairs.
We recall our main
theorem about multi-commodity spanners. 

\thmspanner*

Note that this result is tight in several ways. Firstly, the solution
is also a feasible Steiner forest, and we know that any online Steiner
forest algorithm has cost $\Omega(\log k)$ times the cheapest Steiner
forest, hence we cannot improve the cost even at the expense of
greater stretch. Moreover, there exists a $k$-vertex unweighted graph
with girth $\Omega(\frac{\log k}{\log \log k})$ and size
$\Omega(k \log k)$ \cite{LubotzkyPS88}. Thus, we cannot improve the
stretch by much either.  Finally, $\Omega(k)$ edges are needed to
connect $k$ terminals.



\paragraph{High-level overview.} Our algorithm (Algorithm \ref{alg}) is based on the Berman-Coulston online Steiner forest algorithm. There are two key ingredients to the algorithm. First, it
classify terminals according to the distance to their mates:
$\class(s_i) = \class(t_i) = \lfloor \log_2 d(s_i, t_i) \rfloor$. Let
$Y_j$ be the set of terminals with class at least $j$.
Second, the algorithm maintains a clustering of $Y_j$ as follows. It designates a subset $Z_j
\subseteq Y_j$ that are $2^j/16$-separated from each other as
\emph{centers}, and assigns each $v \in Y_j$ to a \emph{cluster} $C(z)$
where $z$ is the nearest center of $Z_j$ to $u$. 

The rule for adding edges is as follows: if for any $u,v \in Y_j$ with
$d(u,v) \in [2^j, 2^{j+1})$ has $d_H(u,v) > 4\log k\cdot d(u,v)$ (i.e.,
the pair has large stretch in the current graph $H$), the algorithm adds
the edge $(u,v)$, and connects $u$ and $v$ to their respective centers;
we call the former an \emph{augmentation edge} and the latter
\emph{bridge edges} and account for them separately.

Let $A_j$ be the set of augmentation edges $(u,v)$ with $d(u,v) \in
[2^j,2^{j+1})$. First, we show that it suffices to bound the cost and
number of augmentation edges (Lemma \ref{lem:A}). Since the diameter of
each cluster is at most $2^j/8$, each augmentation edge must connect two
different clusters. Thus we can think of a meta-graph where the nodes
correspond to clusters and meta-edges correspond to $A_j$. The bridge
edges allow us to argue that if there is a short path in the meta-graph
between nodes corresponding to clusters $C(z)$ and $C(z')$, then there
is a short path in $H$ between every $u \in C(z)$ and $v \in C(z')$
(Lemma \ref{lem:meta-path}). We then use this to show that there cannot
be a cycle of length less than $\log k$ in the meta-graph and so the
total number of meta-edges is at most $O(1)$ times the number of nodes,
i.e.~$|A_j| \leq O(|Z_j|)$. Finally, we use the separation of $Z_j$ to
show that $2^j|Z_j| = O(\OPT)$ via a ball-packing argument (Lemma
\ref{lem:cost}), and complete the analysis of the algorithm's cost. 
To bound the total number of edges $\sum_j |A_j|$, we use a charging scheme that charges $|A_j|$ to a carefully-chosen set of terminals and show that across all distance scales $j$, the total charge received by each terminal is $O(1)$.

\begin{algorithm}
\caption{Online $(4\log k)$-MCS}
\begin{algorithmic}[1]
  \label{alg}
  \STATE $A_j \leftarrow \emptyset$,
  $B_j \leftarrow \emptyset$, $Z_j \leftarrow \emptyset$ for all $j$;  $H \leftarrow \emptyset$; 
  \WHILE {$(s_i, t_i)$ arrives}
  \STATE Define $\class(s_i) = \class(t_i) = \lfloor \log d(s_i, t_i) \rfloor$
  \FOR {$j = 0$ up to $j = \lfloor \log d(s_i, t_i) \rfloor$}
  \FOR {$u \in \{s_i,t_i\}$}
  \IF {$d(u,z) \geq 2^j/16$ for every center $z \in Z_j$}
  \STATE Add $u$ to $Z_j$
  \STATE Define a new cluster $C(u) = \{u\}$
  \ELSE
  \STATE Choose $z \in Z_j$ closest to $u$ and add $u$ to $C(z)$
  \ENDIF
  \FOR {$u,v$ such that $\min\{\class(u),\class(v)\} \geq j$ and
    $d(u,v) \in [2^j, 2^{j+1})$}
  \IF {$d_H(u,v) > 4\log k \cdot d(u,v)$}
  \STATE Add augmentation edge $(u,v)$ to $A_j$ and $H$ \label{line:aug}
  \STATE Let $z, z' \in Z_j$ such that $u \in C(z)$, $v \in C(z')$
  \STATE Add bridge edges $(u, z)$ and $(v, z')$ to $B_j$ and $H$
  \ENDIF
  \ENDFOR
  \ENDFOR
  \ENDFOR
  \ENDWHILE
\end{algorithmic}
\end{algorithm}

\subsection{Analysis}
To analyze Algorithm \ref{alg}, we use a ``shadow'' algorithm that
builds a family of $O(\log k)$ \emph{meta-graphs} $F_j$, one for each
length scale $2^j$. Whenever the algorithm adds a new cluster---i.e., it
adds a new center $z$ to $Z_j$---the shadow algorithm creates a new node
in $F_j$ corresponding to $C(z)$. We use $C(z)$ to refer to both the
cluster as well as the corresponding node in $F_j$. Whenever Algorithm
\ref{alg} adds an edge $(u,v)$ to $A_j$, the shadow algorithm adds the
\emph{meta-edge} $(C(z),C(z'))$ where $u \in C(z)$ and $v \in C(z')$.
This is well-defined since $u,v$ have class at least $j$ and every
terminal of class at least $j$ belongs to some cluster $C(z)$ for $z \in
Z_j$.  Furthermore, since the diameter of a cluster is less than
$2^j/8$, and each edge of $A_j$ has length at least $2^j$, there are no
self-loops. See Figure \ref{fig:1} for an illustration.

\begin{figure}
  \centering
  \includegraphics[scale=0.3,page=1]{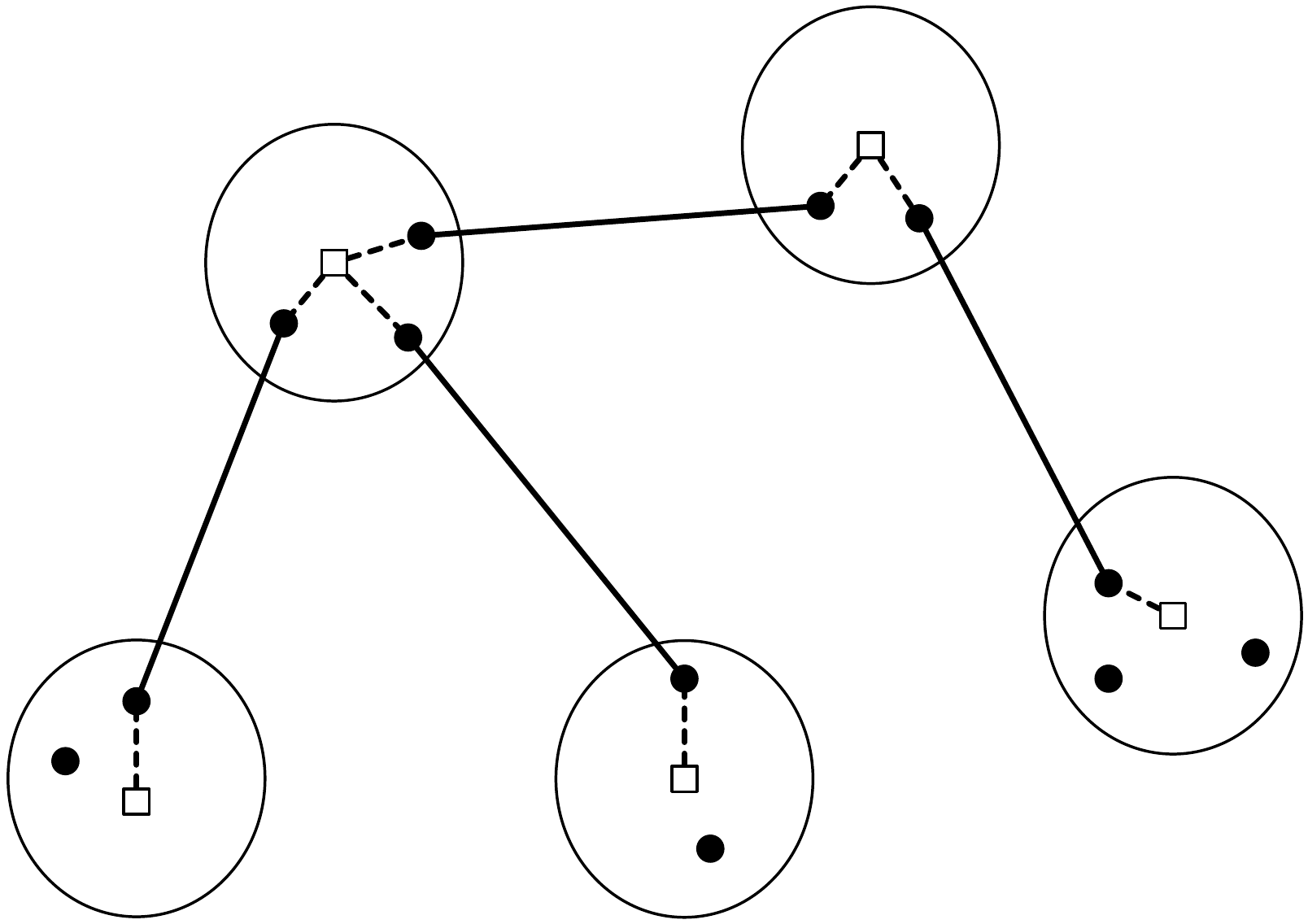}
  \caption{\em An example of the meta-graph. The circles represent clusters
    that are nodes of the meta-graph. The dashed edges and solid edges
    represent bridge and augmentation edges, respectively. Note that
    not all terminals in a cluster has a  bridge edge to the center.} 
  \label{fig:1} \end{figure}

Note that it suffices to consider the augmentation edges to bound the
cost of $H$ as well as the number of edges in $H$. Indeed, every bridge
edge in $B_j$ has a corresponding augmentation edge in $A_j$, and each
augmentation edge corresponds to at most two bridge edges (of no greater
cost). This proves the following lemma.
\begin{lemma}
  \label{lem:A}
  $d(H) \leq \sum_j 2^{j+1}\cdot 3|A_j|$ and $|E(H)| \leq \sum_j 3|A_j|$.
\end{lemma}



The following lemma relates distances in the meta-graph $F_j$ to distances in $H$ and $G$, and is key to bounding $d(H)$ and $|E(H)|$.

\begin{lemma}
  \label{lem:meta-path}
  If there is a path of $\ell$ meta-edges in $F_j$ between
  nodes $C(z)$ and $C(z')$, then
  \begin{enumerate}
  \item   $d_H(z,z') < (3\ell+1)\cdot
  2^j$;
  \item $d(u,v) < 2^j/8 + (3\ell+1)\cdot 2^j$ for every $u \in
  C(z)$ and $v \in C(z')$.
  \end{enumerate}
\end{lemma}

\begin{proof}
  Let the nodes along the path in $F_j$ be $C(z_1), C(z_2), \ldots, C(z_{k+1})$,
  where $z = z_1$ and $z_{k+1} = z'$. Let the corresponding
  augmentation edges be $(u_1, v_1), \ldots, (u_k,v_k)$ where $u_1 \in
  C(z_1), v_k \in C(z_{k+1})$ and $v_i, u_{i+1} \in C(z_i)$ for $1 \leq i
  \leq k$. Note that each
  augmentation edge $(u_i, v_i)$ has two corresponding bridge edges
  $(u_i, z_i)$ and $(v_i, z_{i+1})$.

  We will now show that there is a short path in $H$ from $u$ to $v$
  via these augmentation and bridge edges. Consider the path $P$ in
  $H$ from $z_1$ to $z_{\ell+1}$ consisting of these augmentation
  edges and their corresponding bridge edges. There are $\ell$
  augmentation edges and they have length at most $2^{j+1}$
  each. There are $2\ell$ bridge edges and they have length at most
  $2^j/16$. Therefore, $d(P) \leq 3 \ell \cdot 2^{j+1}$. This proves
  the first part of the lemma.

  We now turn to proving the second part of the lemma.
  Observe that $d(u,v) \leq d(u,z) + d(z,z') + d(z',v)$. Now $d(z,z') \leq d_H(z,z') \leq (3\ell + 1)\cdot 2^j$, from the first part of the lemma. We also have $d(u,z), d(z',v) \leq 2^j/16$ by definition of the clusters. Thus, $d(u,v) \leq 2^j/8 + (3\ell+1)\cdot 2^j$, as desired.
\end{proof}

With this lemma in hand, we can lower-bound the girth of $F_j$.
\begin{lemma}[Girth Bound]
  \label{lem:girth}
  The girth of the meta-graph $F_j$ is at least $\log k$.
\end{lemma}

\begin{proof}
  Suppose, towards a contradiction, that the meta-graph has a cycle of
  length less than $\log k$. Suppose that $(C(z),C(z'))$ is the
  meta-edge of the cycle that was added last, and $(u,v) \in A_j$ is
  the corresponding augmentation edge with $u \in C(z)$ and
  $v \in C(z')$. The edge $(u,v)$ was added during the $j$-th
  iteration of the outer for loop. Consider the state of Algorithm
  \ref{alg} and the meta-graph right before $(u,v)$ was added to $H$
  and $(C(z),C(z'))$ was added to $F_j$. At this time, there is a path
  between nodes $C(z)$ and $C(z')$ in $F_j$ with less than $\log k$
  edges, so by Lemma \ref{lem:meta-path}, we have
  $d_H(u,v) < 4\log k \cdot 2^j \leq 4\log k\cdot d(u,v)$. However,
  this contradicts the condition for adding $(u,v)$ to $A_j$ and so
  $F_j$ has girth at least $\log k$.
\end{proof}

Combining the previous lemma with the following result from extremal graph theory will allow us to bound $|A_j|$ later.
\begin{lemma}~\cite{Bollobas04}
  \label{lem:extremal}
  Every graph on $N$ vertices and girth $k$ has at most
  $O(N^{1+2/(k-1)})$ edges.
\end{lemma}

Let $\OPT$ be the cost of the optimal Steiner forest for the $n$
terminal pairs.

\begin{lemma}[Cost Bound]
  \label{lem:cost}
  $\sum_j 2^{j+1}|A_j| \leq O(1) \sum_j 2^{j+1}|Z_j| \leq O(\log k) \OPT$.
\end{lemma}

\begin{proof}
  By Lemmas \ref{lem:girth} and Lemma \ref{lem:extremal}, we have
  $|A_j| \leq O(1)|Z_j|$. Now we bound $\OPT$ in terms of
  $|Z_j|$. Place a ball of radius $2^j/32$ around each $z \in
  Z_j$.
  Since for any $z,z' \in Z_j$, $d(z,z') \geq 2^j/16$, these balls are
  disjoint. For any $z \in Z_j$, the distance from $z$ and its mate is
  at least $2^j$ so the optimal solution has a path of length at least
  $2^j/32$ that is contained in the ball around $z$. This implies that
  $\OPT \geq (2^j/32) \cdot |Z_j|$.
  Putting all of this together, we have
  $\sum_j 2^{j+1}|A_j| \leq \sum_j 2^{j+1}\cdot O(1)|Z_j| \leq O(\log
  k) \OPT$.
\end{proof}


We now turn to bounding the number of edges $|E(H)| = \sum_j |A_j|$. 
The main idea is to find for each scale $j$, a subset of terminals $X'_j$ such that $|A_j| \leq O(1)|X'_j|$, and each terminal belongs in $X'_j$ for a constant number of scales $j$. 
This allows us to charge $|A_j|$ to the terminals in $X'_j$ for each $j$, and each terminal only receives a charge of $O(1)$ across all $j$. 
Define $\clust_j = \{C(z) : z \in Z_j\}$, the set of clusters corresponding to centers $Z_j$. 
First, we show that there cannot be any large subset of clusters in $\clust_j$ that are far apart from each other. In particular, for any such subset $\clust^+_j \subseteq \clust_j$, $|A_j|$ is at most twice the number of remaining clusters $\clust_j \setminus \clust^+_j$.




\begin{lemma}
  \label{lem:clust}
  Let $\clust^+_j$ be any subset of $\clust_j$ such that for any clusters $C(z), C(z') \in \clust_j$ and terminals $u \in C(z)$ and $v \in C(z')$, we have $d(u,v) \geq 8\cdot 2^j$. Let $\clust^-_j = \clust_j \setminus \clust^+_j$ be the remaining clusters. Then, we have $|A_j| \leq 2 |\clust^-_j|$.
\end{lemma}

\begin{proof}
  Without loss of generality, we assume that the meta-graph $F_j$ is connected (otherwise we can apply the same argument on each component of $F_j$). We begin by showing that $|\clust^+_j| \leq |\clust^-_j|$. The second part of Lemma \ref{lem:meta-path} implies that in $F_j$, if there exists a path consisting of at most $2$ meta-edges between meta-nodes $C(z)$ and $C(z')$, then $d(u,v) < 8 \cdot 2^j$ for all $u \in C(z)$ and $v \in C(z')$. Thus, any path between two meta-nodes $C(z), C(z') \in \clust^+_j$ has to contain at least $3$ meta-edges. This implies that every meta-node in $\clust^+_j$ can only have meta-nodes in $\clust^-_j$ as neighbors, and no two meta-nodes of $\clust^+_j$ can share the same neighbor in $\clust^-_j$. Since $F_j$ is connected, we have that each meta-node in $\clust^+_j$ must have at least one neighbor in $\clust^-_j$. Putting all these facts together gives us $|\clust^+_j| \leq |\clust^-_j|$. 

    By Lemmas \ref{lem:girth} and Lemma \ref{lem:extremal}, we have $|A_j| \leq O(1)|Z_j|$. Since $|Z_j| = |\clust_j|$ and $|\clust^+_j| \leq |\clust^-_j|$, we have $|A_j| \leq 2|\clust^-_j|$.
\end{proof}

\begin{lemma}[Sparsity Bound]
  \label{lem:size}
  $\sum_j |A_j| \leq O(k)$.
\end{lemma}

\begin{proof}
  We can easily construct a family of nets $N_j$, one for each $j$, such that $N_j$ is a $2^j/32$-net and is nesting, i.e. $N_j \subseteq N_{j-1}$ for all $j$. 

First, we claim that each cluster $C(z) \in \clust_j$ contains at least one net point $v \in N_j$. Suppose, towards a contradiction, that there exists a cluster $C(z) \in \clust_j$ such that $C(z) \cap N_j = \emptyset$. Consider a net point $v \in N_j$. Since each terminal is added to the cluster of the closest center in $Z_j$ and $v$ is not assigned to $z$, it must be the case that $d(v,z) \geq 2^j/32$; otherwise, $z$ would have been the closest center to $v$ because $d(z,z') \geq 2^j/16$ for every $z' \in Z_j$. Therefore, every net point $v \in N_j$ has $d(v,z) \geq 2^j/32$, but this contradicts the fact that $N_j$ is a $2^j/32$-net. Thus, each cluster in $\clust_j$ contains at least one net point $v \in N_j$.

Next, for each terminal $v$, we define its \emph{rank} to be the largest $j$ such that $v \in N_j$. Let $N^+_j = N_{j+9}$ and $X'_j = N_j \setminus N^+_j$. Note that because of the nesting property of the nets, the terminals in $X'_j$ are exactly those with rank between $j$ and $j+9$. Let $\clust^+_j$ be the subset of $\clust_j$ that intersects with $N^+_j$ and $\clust^-_j = \clust_j \setminus \clust^+_j$. Now, for every $C(z), C(z') \in \clust^+_j$ and $u \in C(z)$ and $v \in C(z')$, we have $d(u,v) \geq d(z,z') - d(u,z) - d(v,z') \geq 2^{j+9}/32 - 2 \cdot 2^j/16 \geq 2^{j+3}$. Therefore, by Lemma \ref{lem:clust}, we have $|A_j| \leq 2|\clust^-_j|$.

Now, every cluster $C \in \clust^-_j$ contains a terminal of $X'_j$; this is because every cluster contains at least one terminal of $N_j$ and clusters in $\clust^-_j$ do not contain any terminal of $N^+_j$. Since the clusters $\clust^-_j$ are disjoint, we have $|\clust^-_j| \leq |X'_j|$. We now have $\sum_j |A_j| \leq 2 \sum_j |X'_j|$. Observe that every terminal is contained in $X'_j$ for at most a constant number of scales $j$ (each terminal $v \in X'_j$ has $\rank(v) \in [j, j+9]$). Since there are $2k$ terminals, we have $\sum_j |X'_j| \leq O(k)$ and so, $\sum_j |A_j| \leq O(k)$.
\end{proof}

Lemmas \ref{lem:A}, \ref{lem:cost} and \ref{lem:size} give us Theorem
\ref{thm:spanner}.


\paragraph{Handling lack of knowledge of $k$.}
Our Algorithm \ref{alg} requires the knowledge of the value of $k$ in
determining when augmentation edges must be added, which is not
available a priori in the online setting. However, it is relatively
straightforward to get over this by noting that the subgraph $H$
maintained by the algorithm can be decomposed into one set of edges for
every distance scale $j$. At a specific distance scale, the number of
centers in $Z_j$ gives us a good lower bound on $k_j$ which can be used
in place of $k$ in determining the stretch condition: in other words, we
use $4\log k_j$ in the stretch condition when processing the loop for
$j$, and update it whenever the number of clusters in $Z_j$ doubles.


\section{Randomized Oblivious Buy at Bulk}
\label{sec:oblivious-bab-rand}

In this section, we describe our randomized algorithm for the oblivious
case, where we don't know the buy-at-bulk function in advance, and hence
have to give a solution that is simultaneously good for all buy-at-bulk
functions. Let us recall our main result for this section,
Theorem~\ref{thm:ob-BAB-rand}:

\BabUnknownRand*

Recall that this guarantee can also be achieved by performing Bartal's
tree embeddings, but the ideas we develop here will be crucial to
getting the deterministic algorithm in \S\ref{sec:oblivious-bab-det}.



For this section and the next, we use $g_i$ to denote the
``rent-or-buy'' concave function $g_i(x) = \min \{2^i, x\}$ and write
$\OPT_i$ to be the cost of the optimal solution for $g_i$. Given a
routing solution $P$, we denote by $\cost_i(P)$ the cost of $P$ under
$g_i$. By Lemma~\ref{lem:basis} it suffices to show that simultaneously
for all $g_i$,
\[  \cost_i(P) \leq O(\log^2 k)  \OPT_i. \]
This is what we show in the rest of this section and the next.




\subsection{High-Level Description}
We begin by reviewing the main ideas for approximating a fixed
rent-or-buy function $g_i = \min\{x, 2^i\}$ online. Essentially, it
boils down to choosing a good set of \emph{``buy'' terminals} $X_i$
(which includes the root) online. The terminal set $X_i$ is connected by
buying the edges of an online Steiner tree $F_i$ on them. The remaining
terminals are called \emph{``rent'' terminals} and are routed as
follows: For each rent terminal $u$, let $X_i(v)$ be the set of bought
terminals at the time of its arrival; the terminal $u$ is then routed to
the nearest buy terminal $v \in X_i(v)$ (the edge $(u,v)$ is rented),
and then via $v$'s path to the root in $F_i$. Observe that the cost of
buying edges is $2^i d(F_i)$ and the cost of renting edges is $\sum_{v
  \notin X_i} d(v, X_i(v))$. A key ingredient\footnote{We will actually 
need a stronger version, Lemma \ref{lem:opt}} (see, e.g.,~\cite{AwerbuchAB04}) 
is the fact that if we choose each terminal $v$ to be a buy terminal with probability $2^{-i}$, then
\begin{equation}
2^i d(\Steiner(X_i)) + \sum_{v \notin X_i} d(v,X_i(v)) \leq O(\log k)
\OPT_i.\label{eq:ROB-rand}
\end{equation}

But this approach requires knowing the particular function $g_i$,
whereas we need to oblivious, i.e., to be good for all $g_i$
simultaneously. This suggests the following approach to designing an
online oblivious algorithm. Firstly, we assign a \emph{type} $\type(v)$ to
each terminal $v$ at random: $\type(v) = i$ with probability
$2^{-i}$. As before, we assign the root $r$ to have type $\infty$. Thus,
if we define $X_i$ to be the set of terminals with type at least $i$,
then it satisfies Inequality
\eqref{eq:ROB-rand}. 

Secondly, for each type $i$, we maintain an online spanner $F_i$ on
$X_i$ using the algorithm from \S\ref{sec:spanner}. We think of the
edges that belong to a spanner of type at least $i$ as the edges that
our solution buys if it is given the $i$-th rent-or-buy function
$g_i$. Then, each terminal $v$ is routed using a path $P(v)$ that goes
through spanners of progressively higher type until we reach the
root. Let $P_i(v)$ be the prefix of $P(v)$ right before it enters a
spanner of type at least $i$. This is the part of $P(v)$ that will be
rented for $g_i$. At a high level, the lightness and low distortion
properties of the spanners will allow us to prove that the routing $P$
satisfies the following bounded cost properties for all $i$:
\begin{itemize}
\item Bounded buy cost: $2^id(F_i) \leq O(\log k)\Steiner(X_i)$;
\item Bounded rent cost:
  $\sum_{v : \type(v) < i} d(P_i(v)) \leq O(\log k) d(v,X_i(v))$,
  i.e. the distance from $v$ to $X_i(v)$ on the path $P(v)$ is not
  much larger than the actual distance.
\end{itemize}


Some care needs to be taken in choosing which spanners to route
through. Fix a terminal $v$, and for each $i \geq \type(v)$, let $w_i$
be the nearest terminal of type at least $i$ and call it the $i$-th
\emph{waypoint} for $v$. A natural idea is to route from $v$ to
$w_{\type(v)+1}$ using the spanner $F_{\type(v)}$, then from
$w_{\type(v)+1}$ to $w_{\type(v)+2}$ using $F_{\type(v)+1}$, and so
on, until we reach the root. However, the problem is that even if we
were to route between consecutive waypoints $w_i, w_{i+1}$ using the
direct edge $(w_i, w_{i+1})$, the total length of the path before we
can reach a given $w_{i'}$ might be too large. To overcome this issue,
we bypass waypoints that are of roughly the same distance from $v$ and
route through a subset of waypoints whose distances from $v$
are geometrically increasing. See Algorithm \ref{alg:rand-ob} for a
formal description of our algorithm.

\begin{figure}
  \centering
  \includegraphics[scale=0.5]{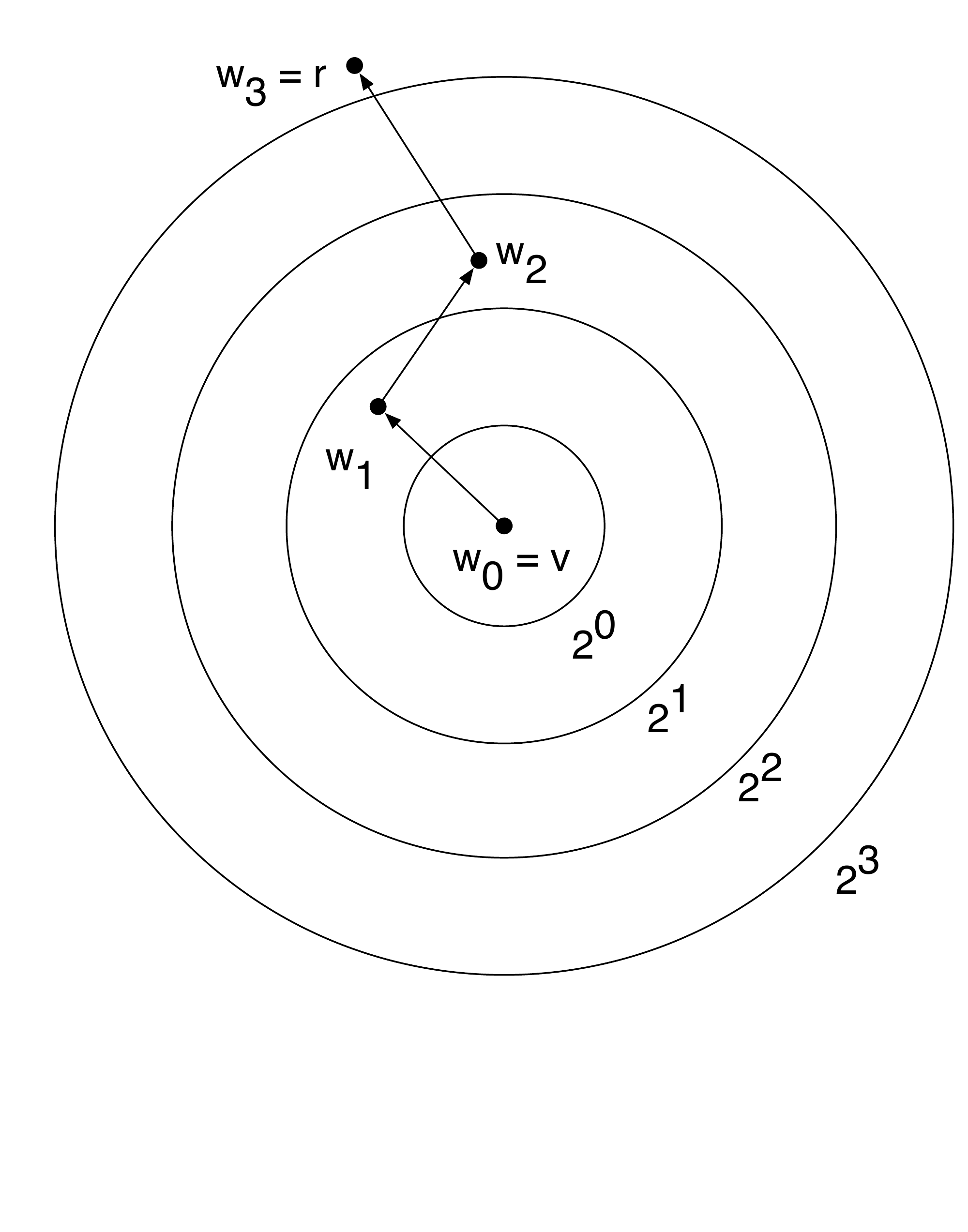}
  \caption{The arrow from $w_{i}$ to $w_{i+1}$ represents the shortest
    path between $w_{i}(v)$ and $w_{i+1}(v)$ in $F_{\tau(w_i)}$. Note
    that it is possible for a ring (in this figure, the fourth ring)
    to not have a waypoint.} \label{fig:path} \end{figure}


\begin{algorithm}
\caption{Randomized Online Oblivious Algorithm for Buy-at-Bulk} 
\begin{algorithmic}[1]
  \label{alg:rand-ob}
  \STATE Initialize $X_i \leftarrow \emptyset, F_i \leftarrow \emptyset$ for all
  $i$
  \STATE Set $\type(r) \leftarrow \infty$
  \WHILE {terminal $v$ arrives}
  \STATE Choose a random type $\type(v)$, with $\type(v) = i$ w.p. $2^{-i}$
  \STATE For $i \leq \type(v)$, add $v$ to $X_i$ and update spanner $F_i$ to include $v$
  \STATE For $i > \type(v)$, define waypoint $w_i$ to be terminal of $X_i$ nearest to $v$
  \STATE Set initial waypoint to be $v$
  \FOR {$j = 0$ up to $\lfloor \log d(v,r) \rfloor$}
  \STATE Define the ring $R_j = \{w_i : d(v,w_i) \in [2^j, 2^{j+1})\}$
  \IF {$R_j \neq \emptyset$}
  \STATE Choose waypoint $w_i$ in $R_j$ of highest type and route from current waypoint $w_{i'}$ to $w_i$ along shortest path in $F_{i'}$
  \ENDIF
  \ENDFOR
  \ENDWHILE
\end{algorithmic}
\end{algorithm}

\subsection{Analysis}

We now analyze $\cost_i(P)$. For each
terminal $v$, define $X_i(v)$ to be the subset of terminals of type
at least $i$ when $v$ arrived, excluding $v$ itself. 

\begin{lemma}
  \label{lem:cost-decomp}
  $\cost_i(P) \leq 2^i \sum_{i' \geq i} d(F_{i'}) + \sum_{v : \type(v) < i} d(P_i(v))$.
\end{lemma}

\begin{proof}
  Recall that
  $\cost_i(P) = \sum_{e \in E} d(e)\min\{2^i, \load_P(e)\}$, where
  $\load_P(e)$ is the total number of terminals $v$ whose paths $P(v)$
  include $e$.  We will divide the edges used by the routing into two
  sets $\rent_i$ and $\buy_i$, and think of $\rent_i$ as the set of rented edges
  (i.e. we pay $d(e)\load_P(e)$ on edges $e \in \rent_i$), and $\buy_i$ as
  the set of bought edges (i.e. we pay $2^id(e)$ per edge in
  $\buy_i$). Define $\rent_i = \bigcup_v P_i(v)$ and
  $\buy_i = \bigcup_v P(v) \setminus P_i(v)$. We have
\[\cost_i(P) = \sum_{e \in E} d(e)\min\{2^i, \load_P(e)\} \leq 2^i
d(\buy_i) + \sum_{e \in \rent_i} d(e)\load_P(e).\]

Since $P(v) \setminus P_i(v) \subseteq \bigcup_{i' \geq i} F_{i'}$, we
have that $2^i d(\buy_i) \leq 2^i \sum_{i' \geq i} d(F_{i'})$. Now
observe that for each $e \in \rent_i$, $\load_P(e) $ is exactly the number
of terminals $v$ of type less than $i$ with $P_i(v) \ni e$. This is
because $P_i(v) \neq \emptyset$ only for terminals $v$ of type less
than $i$. Thus, we have
$\sum_{e \in \rent_i} d(e) \load_P(e) = \sum_{v : \type(v) < i}
d(P_i(v))$.
\end{proof}

We define $\Steiner(S)$ to be the cost of the optimal Steiner tree for
terminals $S$.

\begin{lemma}
  \label{lem:opt}
  We have that
  \begin{enumerate}
  \item $E[2^i \Steiner(X_i)] \leq \OPT_i$ and
  \item $E[\sum_v d(v, X_i(v))] \leq O(\log k)\OPT_i$.
  \end{enumerate}
\end{lemma}

\begin{proof}
  We begin with the first part of the lemma. For terminal $v \in X$,
  let $P^*_i(v)$ be the path to the root in $\OPT_i$. Define $\buy^*_i$ to
  be the set of edges over which at least $2^i$ terminals
  are routed, and for each $v$, define $\rent^*_i(v) = P^*_i(v) \setminus
  \buy^*_i$. The cost of the optimal solution is
  \[\OPT_i = 2^i d(\buy^*_i) + \sum_{v \in X} d(\rent^*_i(v)).\]

  Define $H_i = \bigcup_{v \in X_i} P^*_i(v)$. Since $H_i$
  is a Steiner tree for terminals $X_i$,
  \begin{align*}
    E[2^i \Steiner(X_i)]
    &\leq 2^i d(\buy^*_i) + E\bigg[2^i \sum_{v \in X_i} d(\rent^*_i(v))\bigg]\\
    &\leq 2^i d(\buy^*_i) + 2^i \sum_v d(\rent^*_i(v))\cdot \Pr[v \in X_i]\\    
    &\leq 2^i d(\buy^*_i) + \sum_v d(\rent^*_i(v)) = \OPT_i.
  \end{align*}
  This proves the first part of the lemma.

  We now turn to the second part of the lemma. For this, we will show
  that
  $E\bigg[2^i \sum_{v \in X_i} d(v, X_i(v))\bigg] =
  E\bigg[\sum_v d(v, X_i(v))\bigg]$,
  then prove
  $\sum_{v \in X_i} d(v, X_i(v)) \leq O(\log
  k)\Steiner(X_i)$,
  and apply the first part of the lemma. By linearity of expectation,
  \begin{align*}
    E\bigg[2^i \sum_{v \in X_i} d(v, X_i(v))\bigg]
    &= 2^i \sum_v E[d(v, X_i(v)) \cdot 1\{v \in X_i\}]\\
    &= 2^i \sum_v E[d(v, X_i(v))] \cdot \Pr[v \in X_i]\\
    &= E\bigg[\sum_v d(v, X_i(v))\bigg]
  \end{align*}
  where the second equality follows from the fact that whether $v \in
  X_i$ is independent of $d(v,X_i(v))$, and the last equality from the fact that $\Pr[v \in X_i] = 1/2^i$.

  Next, we prove
  $\sum_{v \in X_i} d(v, X_i(v)) \leq O(\log
  k)\Steiner(X_i)$.
  Consider the online Steiner tree instance defined on terminals
  $X_i$. The cost of the greedy algorithm, which connects each
  terminal to the nearest previous terminal, is exactly
  $\sum_{v \in X_i} d(v, X_i(v))$. The fact that the greedy
  algorithm is $O(\log k)$-competitive for the online Steiner tree
  problem implies
  $\sum_{v \in X_i} d(v, X_i(v)) \leq O(\log k)
  \Steiner(X_i)$. Therefore,
  \[E\bigg[\sum_v d(v, X_i(v))\bigg] = E\bigg[2^i \sum_{v \in X_i} d(v,
  X_i(v))\bigg] \leq O(\log k) E[2^i \Steiner(X_i)] \leq O(\log k)
  \OPT_i.\] This completes the proof of the second part of the lemma.
\end{proof}

In the following, denote by $P_i(v)$ the prefix of $P(v)$ right before
it enters a spanner $F_{i'}$ with $i' \geq i$.

\begin{lemma}
  \label{lem:dist}
  For every $i > \type(v)$, we have $d(P_i(v)) \leq O(\log k) d(v, X_i(v))$.
\end{lemma}

\begin{proof}
  Note that $w_i$ is the nearest terminal of $X_i(v)$ so $d(v,w_i) = d(v,X_i(v))$. Define $W$ to be the subset of waypoints $w_i$ such that $w_i$ is the waypoint of highest type in some ring $R_j$, and $I = \{i : w_i \in W\}$. For every type $i' > \type(v)$ that is not in $I$, there exists $i \in I$ such that $i > i'$ and $d(v, w_{i}) \leq 2d(v,w_{i'})$; moreover, since $i > i'$, we have $d(P_{i'}(v)) \leq d(P_{i}(v))$. Therefore, it suffices to prove $d(P_i(v)) \leq O(\log k) d(v,w_i)$ for types $i \in I$. 

For brevity, we abuse notation and write $w_\ell$ to denote the waypoint in $W$ with $\ell$-th highest type; $F_\ell$ to denote the spanner of type $\type(w_\ell)$. Let $P_\ell(v)$ be the prefix of $P(v)$ ending at $w_\ell$. 
We now show that $d(P_\ell(v)) \leq O(\log k)d(v,w_\ell)$.

Since $P_\ell(v)$ consists of the union of the shortest path between $w_{\ell'-1}$ and $w_{\ell'}$ in $F_{\ell'-1}$, for every $\ell' < \ell$, the length of $P_\ell(v)$ is 
$d(P_\ell(v)) = \sum_{\ell' \leq \ell}d_{F_{\ell'-1}}(w_{\ell'-1}, w_{\ell'})$.
For every $\ell' < \ell$, since $F_{\ell'-1}$ is a spanner with stretch $O(\log k)$, we have 
\[d_{F_{\ell'-1}}(w_{\ell'-1}, w_{\ell'}) \leq O(\log k)d(w_{\ell'-1}, w_{\ell'}) \leq O(\log k)(d(v, w_{\ell'-1}) + d(v, w_{\ell'})).\]
Thus, 
  \begin{align*}
    d(P_\ell(v))
    &\leq O(\log k)\sum_{\ell' < \ell} (d(v, w_{\ell'-1}) + d(v, w_{\ell'})) \leq O(\log k) d(v,w_\ell),
  \end{align*}
where the last inequality follows from the fact that we chose $W$ to contain at most one waypoint per ring and the rings have geometrically increasing distances.
\end{proof}

\begin{lemma}
  $E[\cost_i(P)] = O(\log^2 k)\OPT_i$.
\end{lemma}

\begin{proof}
  By Lemma \ref{lem:cost-decomp}, the expected cost of the algorithm's solution is
    $E[\cost_i(P)] \leq E\big[2^i \sum_{i' \geq i} d(F_{i'})\big] + E\big[\sum_{v : \type(v) < i} d(P_i(v))\big]$.
  Let $i' > i$. Since $d(F_{i'})$ is a $(\log k)$-competitive $(\log k)$-spanner on
  $X_i$, we have $d(F_{i'}) \leq O(\log k)\Steiner(X_{i'})$. Together
  with the fact that $\Steiner(X_{i'}) \leq \Steiner(X_i)$ (since
  $X_{i'}$ contains $X_i$), we
  have \[E\big[2^i \sum_{i' \geq i} d(F_{i'})\big] \leq O(\log k) \sum_{i' \geq i}
  E[2^i \Steiner(X_{i'})] \leq O(\log^2 k) E[2^i \Steiner(X_i)] \leq
  O(\log^2 k) \OPT_i,\]
  where the last inequality follows from Lemma \ref{lem:opt}.

  Next, we turn to bounding $E\big[\sum_{v : \type(v) < i} d(P_i(v))\big]$. By
  Lemmas \ref{lem:dist} and \ref{lem:opt}, we have
  \[E\bigg[\sum_{v : \type(v) < i} d(P_i(v))\bigg] \leq O(\log k)
  E\bigg[\sum_{\type(v) < i} d(v, X_i(v))\bigg] \leq O(\log^2 k) \OPT_i.\]
\end{proof}


\section{Deterministic Oblivious Buy at Bulk}
\label{sec:oblivious-bab-det}


In this section, we present our deterministic algorithm for the
oblivious setting, and prove Theorem \ref{thm:ob-BAB-det}.

\BabUnknownDet*

Our deterministic oblivious algorithm is based on the randomized version
from \S\ref{sec:oblivious-bab-rand} (Algorithm \ref{alg:rand-ob}) with
two modifications. First, we remove the need for randomness by using the
deterministic type selection rule used by our non-oblivious algorithm
from \S\ref{sec:non-obl} (Algorithm \ref{alg:main}). We can show that
this modification already yields an $O(\log^3 k)$-competitive
algorithm. To get the competitive ratio down to $O(\log^{2.5} k)$, we
need to use a different routing scheme. For clarity of presentation, we
first present the simpler $O(\log^3 k)$-competitive algorithm and later
(in Section \ref{sec:det-oblivious-improved}) show how to modify it to
obtain the $O(\log^{2.5} k)$-competitive algorithm of Theorem
\ref{thm:ob-BAB-det}.

\subsection{An $O(\log^3 k)$-competitive Algorithm}

We will assign a type $\type(v)$ to each terminal $v$, with the root having type $\type(r) = \infty$. As before, we use $X_i$ to denote the set of terminals with type at least $i$. We will also maintain a $(\log k, \log k)$-spanner $F_i$ on $X_i$. 

\begin{algorithm}
\caption{Simpler Online Deterministic Oblivious Algorithm for Buy-at-Bulk} 
\begin{algorithmic}[1]
  \label{alg:det-ob}
  \STATE Initialize $X_i \leftarrow \emptyset, F_i \leftarrow \emptyset$ for all
  $i$
  \STATE Set $\type(r) \leftarrow \infty$
  \WHILE {terminal $v$ arrives}
  \FOR {each type $i$}
  \STATE Define $\class_i(v)$ to be the highest $j$ such that $v$ is a level $j$ center in spanner $F_i$ if it were to be included in $F_i$
  \STATE Define the ball $B_i(v) = \{u : d(u,v) \leq 2^{\class_i(v)}/2^3\}$. 
  \ENDFOR
  \STATE Assign type $\type(v) = \max \left\{i : |B_i(v)| \geq 2^i \right\}$.
  \STATE For $i \leq \type(v)$, update spanner $F_i$ to include $v$
  \STATE For $i > \type(v)$, define waypoint $w_i$ to be terminal of $X_i$ nearest to $v$
  \STATE Set initial waypoint to be $v$
  \FOR {$j = 0$ up to $\lfloor \log d(v,r) \rfloor$}
  \STATE Define the ring $R_j = \{w_i : d(v,w_i) \in [2^j, 2^{j+1})\}$
  \IF {$R_j \neq \emptyset$}
  \STATE Choose waypoint $w_i$ in $R_j$ of highest type and route from current waypoint $w_{i'}$ to $w_i$ along shortest path in $F_{i'}$
  \ENDIF
  \ENDFOR
  \ENDWHILE
\end{algorithmic}
\end{algorithm}


\subsubsection{Analysis} 
Fix a type $i$. We will now show that $\cost_i(P) \leq O(\log^2 k)\OPT_i$. As before, define $P_i(v)$ to be the segment of $P(v)$ right before it enters a spanner $F_{i'}$ with $i' \geq i$. (Note that $P_i(v) = \emptyset$ if $\type(v) \geq i$.) 

By Lemma \ref{lem:cost-decomp}, we have $\cost_i(P) \leq 2^i \sum_{i' \geq i} d(F_i) + \sum_{v : \type(v) < i} d(P_i(v))$. Next, we bound $2^i \sum_{i' \geq i} d(F_i)$ and $\sum_{v : \type(v) < i} d(P_i(v))$ separately. 

\begin{lemma}
  $\sum_{v : \type(v) < i} d(P_i(v)) \leq O(\log^2 k)\OPT_i$.
\end{lemma}

\begin{proof}
  Using the same argument as in the proof of Lemma \ref{lem:dist} gives us that for every terminal $v$, and type $i > \type(v)$, we have $d(P_i(v)) \leq O(\log k) d(v, X_i(v))$. Thus, \[\sum_{v : \type(v) < i} d(P_i(v)) \leq O(\log k) \sum_{v : \type(v) < i} d(v,X_i(v)).\] Using a similar proof as that of Lemma \ref{lem:rent-cost}, we also have \[\sum_{v : \type(v) < i} d(v,X_i(v)) \leq O(\log k)\OPT_i.\] Combining these two inequalities gives us the lemma.
\end{proof}

\begin{lemma}
  \label{lem:det-ob-buy-1}
  $2^i \sum_{i' \geq i} d(F_{i'}) \leq O(\log^2 k)\OPT_i(T)$ for any HST embedding $T$.
\end{lemma}

\begin{proof}
  Let $i' > i$. By Lemma \ref{lem:cost} and the fact that $\class_{i'}(v)$ is the highest level $j$ such that $v$ is a level-$j$ center in $F_{i'}$, we have that $2^i d(F_{i'}) \leq O(1) 2^i \sum_{v \in X_{i'}} 2^{\class_{i'}(v)}$. Now, we would like to say that the RHS of the previous inequality is at most $O(1)\OPT_{i'}(T)$ for any HST embedding $T$. We can prove this using a similar proof as in Lemma \ref{lem:fixed-cost} if for every $v$ with $\type(v) \geq i'$, we have $|B_{i'}(v)| \geq 2^{i'}$. Unfortunately, this is not necessarily the case. A terminal $v$ might have $\type(v) \geq i'$, i.e.~it has at least $2^{\type(v)} \geq 2^{i'}$ terminals in its $\type(v)$-th ball $B_{\type(v)}(v)$, but its $i'$-th ball $B_{i'}$ might have a radius much smaller than $B_{\type(v)}(v)$ (i.e. $\class_{i'}(v) < \class_{\type(v)}(v)$), and thus does not contain at least $2^{i'}$ terminals. In this case, roughly speaking, we will charge the cost of including $v$ in $F_{i'}$ to the cost of including it in $F_{\type(v)}$.

  More formally, for each type $i'$, define $A_{i'} = \{v \in X_{i'} : |B_{i'}(v)| \geq 2^{i'}\}$, i.e. this is the set of terminals who can ``pay'' to be included in spanner $F_{i'}$. 

  \begin{claim}
    Let $v$ be a terminal with type $\type(v) > i'$. If $v \notin A_{i'}$, then $\class_{i'}(v) < \class_{\type(v)}(v)$.
  \end{claim}

  \begin{proof}
    The balls $B_{\type(v)}(v)$ and $B_{i'}(v)$ have radii $2^{\class_{\type(v)}(v)-3}$ and $2^{\class_{i'}(v)-3}$, respectively. By definition of types, $v \in A_{\type(v)}$ and so $|B_{\type(v)}(v)| \geq 2^{\type(v)} \geq 2^{i'}$. Since $v \notin A_{i'}$, the radius of $B_{i'}(v)$ must be smaller than the radius of $B_{\type(v)}(v)$ and so $\class_{i'}(v) < \class_{\type(v)}(v)$.
  \end{proof}

  Define $N(v) = \{i' < \type(v) : v \in X_{i'} \setminus A_{i'}\}$, i.e.~the set of types $i'$ for which $v$ is included in spanner $F_{i'}$ but could not pay to be included in $F_{i'}$. Combining this with the above claim, we get
  \begin{align*}
    \sum_{i' \geq i} d(F_{i'})
    &\leq O(1)\sum_{i' \geq i} \left(\sum_{v \in A_{i'}} 2^{\class_{i'}(v)} + \sum_{v \in X_{i'} \setminus A_{i'}} 2^{\class_{\type(v)}(v)}\right) \\
    &= O(1)\sum_{i' \geq i} \sum_{v \in A_{i'}} 2^{\class_{i'}(v)} + \sum_{v \in X_i} |N(v)|\cdot 2^{\class_{\type(v)}(v)} \\
    &\leq O(1)\sum_{i' \geq i} \sum_{v \in A_{i'}} 2^{\class_{i'}(v)} + O(\log k)\sum_{v \in X_i} 2^{\class_{\type(v)}(v)} \\
    &\leq O(\log k)\sum_{i' \geq i} \sum_{v \in A_{i'}} 2^{\class_{i'}(v)}
  \end{align*}
  where the last inequality follows from the fact that $v \in A_{\class_{\type(v)}(v)}$ and $\type(v) \geq i$.

  Using a similar proof as in Lemma \ref{lem:fixed-cost}, we have $2^i \sum_{v \in A_{i'}} 2^{\class_{i'}(v)} \leq O(1)\OPT_i(T)$. Thus,
  \[2^i \sum_{i' \geq i} d(F_{i'}) \leq O(\log^2 k)\OPT_i(T),\]
  as desired.
\end{proof}

With these two lemmas and Lemma \ref{lem:decomp} in hand, we have that $\ALG \leq O(\log^3 k)\OPT_i$ and thus the algorithm has competitive ratio $O(\log^3 k)$.

\subsection{Improving to $O(\log^{2.5} k )$}
\label{sec:det-oblivious-improved}

We now describe the modified routing scheme. Suppose we know the number of terminals $k$ of the online instance in advance (we will show how to remove this assumption at the end). Let $L = \sqrt{\log k}$. The main idea is that instead of including each terminal $v$ in spanner $F_i$ for every $i \leq \type(v)$, we only include $v$ in $F_i$ for $\type(v) - L \leq i \leq \type(v)$. This helps us bound the number of spanners that $v$ is included but could not pay for by $L$. However, we now cannot reach all terminals of type higher than $i$ using spanner $F_i$. To cope with this, we will group the $L^2$ different types into consecutive intervals of length $L$. The routing for a terminal $v$ now proceeds in several phases, one per interval, starting from the interval containing $\type(v)$. During phase $t$ (corresponding to the interval $[tL, (t+1)L]$), we find a path $Q_t$ from the waypoint $w_{tL}$ (or $v$ if the interval contains $\type(v)$) to $w_{(t+1)L}$ by passing through intermediate waypoints that belong to rings geometrically increasing distances.

In the following, we write $u \in F_i(v)$ to mean that $u$ is a terminal included in spanner $F_i$ before $v$ arrived.
\begin{algorithm}
\caption{Improved Online Deterministic Oblivious Algorithm for Buy-at-Bulk} 
\begin{algorithmic}[1]
  \label{alg:det-ob-improved}
  \STATE Initialize $X_i \leftarrow \emptyset, F_i \leftarrow \emptyset$ for all
  $i$
  \STATE Set $\type(r) \leftarrow \infty$
  \WHILE {terminal $v$ arrives}
  \FOR {each type $i$}
  \STATE Define $\class_i(v)$ to be the highest $j$ such that $v$ is a level $j$ center in spanner $F_i$ if it were to be included in $F_i$
  \STATE Define the ball $B_i(v) = \{u : d(u,v) \leq 2^{\class_i(v)}/2^3\}$. 
  \ENDFOR
  \STATE Assign type $\type(v) = \max \left\{i : |B_i(v)| \geq 2^i \right\}$.
  \STATE For $\type(v) - L \leq i \leq \type(v)$, update spanner $F_i$ to include $v$
  \STATE For $i > \type(v)$, define waypoint $w_i$ to be terminal of $F_i$ nearest to $v$
  \STATE Set initial waypoint to be $v$
  \FOR {$t$ such that $\type(v) - 1 \leq tL \leq \log k$}
  \STATE Define the ring $R_j(t) = \{w_i : tL \leq i \leq (t+1)L \wedge d(v,w_i) \in [2^j, 2^{j+1})\}$
  \STATE Initialize $Q_t \leftarrow \emptyset$
  \WHILE {current waypoint is not $w_{(t+1)L}$}
  \IF {$R_j(t) \neq \emptyset$}
  \STATE Choose waypoint $w_i$ in $R_j(t)$ of highest type and route from current waypoint $w_{i'}$ to $w_i$ along shortest path in $F_{i'}$
  \STATE Add shortest path to $Q_t$
  \ENDIF
  \ENDWHILE
  \ENDFOR
  \ENDWHILE
\end{algorithmic}
\end{algorithm}


\begin{lemma}
  \label{lem:det-ob-dist}
  For each terminal $v$, we have 
  \[d(P_i(v)) \leq O(\log k) \left(\sum_{t : \type(v) < (t+1)L \leq i} d(v, F_{(t+1)L}(v)) + d(v, F_i(v))\right).\]
\end{lemma}
  
\begin{proof}
  Let $Q_t$ be the segment of $P(v)$ starting from $w_{tL}(v)$ to $w_{(t+1)L}(v)$ and $Q_t(i)$ be the segment of $Q_t$ right before it enters a spanner $F_{i'}$ with $i' \geq i$. Using a similar proof as in Lemma \ref{lem:dist}, for every $t$ and $i \in [tL, (t+1)L]$, we get that $d(Q_t(i)) \leq O(\log k) d(v, w_i(v))$ and $d(Q_t) \leq O(\log k) d(v, w_{(t+1)L}(v))$. For each type $i$, define $t_i$ to be the integer such that $i \in [t_iL, (t_i+1)L)$. We have
  \[d(P_i(v)) = \sum_{t : \type(v) < (t+1)L < i} d(Q_t) + d(Q_{t_i}(i)) \leq O(\log k)\left(\sum_{t : \type(v) < (t+1)L < i} d(v, w_{(t+1)L}(v)) + d(v, w_i(v))\right).\]
\end{proof}

\begin{lemma}
  $\sum_{v: \type(v) < i}d(P_i(v)) \leq O(\log^{2.5} k)\OPT_i$.
\end{lemma}

\begin{proof}
  Using a similar proof as Lemma \ref{lem:rent-cost}, we get that $\sum_{v : \type(v) < i'} d(v, F_{i'}(v)) \leq O(\log k) \OPT_{i'}$. Combining this with Lemma \ref{lem:det-ob-dist} gives us
  \begin{align*}
    \sum_{v : \type(v) < i} d(P_i(v)) 
    &\leq O(\log k) \left(\sum_{t : (t+1)L < i} \sum_{v : \type(v) < (t+1)L} d(v, F_{(t+1)L}(v)) + \sum_{v : \type(v) < i} d(v, F_i(v))\right) \\
    &\leq O(\log^2 k) \left(\sum_{t : (t+1)L < i} \OPT_{(t+1)L} + \OPT_i\right) \\
    &\leq O(\log^{2.5} k) \OPT_i,
  \end{align*}
  where the last inequality follows from the fact that $\OPT_{i'} \leq \OPT_i$ if $i' < i$ and that the number of integers $t$ such that $(t+1)L < i$ is at most $L = \sqrt{\log k}$.
\end{proof}

\begin{lemma}
  $2^i \sum_{i' \geq i} d(F_i) \leq O(\log^{1.5} k)\OPT_i(T)$ for any HST embedding $T$.
\end{lemma}

\begin{proof}
  The proof of this lemma is similar to that of Lemma \ref{lem:det-ob-buy-1} except that the total number of types $i'$ such that $v$ is included in spanner $F_{i'}$ is at most $L = \sqrt{\log k}$. 
\end{proof}

\paragraph{Removing the need to know $k$ in advance.}
First, instead of using intervals of length $\sqrt{\log k}$, we will use intervals between perfect squares. In particular, the $t$-th interval is $[t^2, (t+1)^2]$. Second, each terminal $v$ is included in spanners $F_i$ where $\type(v) \geq i \geq \type(v) - \sqrt{i}$.




\section{Conclusions}

Several open questions remain: can we do better than $O(\log^2 k)$ for the oblivious randomized case, and $O(\log^{2.5} k)$ for the oblivious deterministic case? Can we get online constructions of tree embeddings with stretch better than $O(\log k \log \Delta)$? 


\bibliographystyle{alpha}
{\small \bibliography{BAB}}





\end{document}